\newif\ifIEEE
\IEEEtrue
\ifIEEE
   \documentclass[journal,final,letterpaper]{IEEEtran}
   \newcommand{\bibauthor}[1]{#1}
   \newcommand{\bibpaper}[1]{``#1''}
   
\else
   \documentclass[12pt]{article}
   \usepackage{amsthm}
   \newcommand{\bibauthor}[1]{\textsc{#1}}
   \newcommand{\bibpaper}[1]{\textsl{#1}}
   \newenvironment{IEEEkeywords}{\begin{small}%
                                  \textbf{Index Terms} ---}{\end{small}}
\fi
\usepackage{amsfonts,bm}
\usepackage{amsthm}
\usepackage{color}
\usepackage{pict2e}
\newcommand{\bibbook}[1]{\textit{#1}}

\newcommand{\bibperiodical}[1]{\textit{#1}}

\newtheorem{theorem}{Theorem}
\newtheorem{proposition}[theorem]{Proposition}
\newtheorem{lemma}[theorem]{Lemma}
\newtheorem{corollary}[theorem]{Corollary}

\theoremstyle{remark}
\newtheorem{remark}{Remark}
\theoremstyle{definition}
\newtheorem{example}{Example}
\renewcommand{\mathbf}[1]{{\bm{#1}}}     

\newcommand{\encoder}{{\mathcal{E}}}
\newcommand{\bldzero}{{\mathbf{0}}}
\newcommand{\bldone}{{\mathbf{1}}}
\newcommand{\bldc}{{\mathbf{c}}}

\newcommand{\bldx}{{\mathbf{x}}}
\newcommand{\bldy}{{\mathbf{y}}}
\newcommand{\bldw}{{\mathbf{w}}}
\newcommand{\bldxi}{{\mathbf{\xi}}}
\newcommand{\anticipation}{{\mathcal{A}}}

\newcommand{\FF}{{\mathcal{F}}}
\newcommand{\XX}{{\mathcal{X}}}
\newcommand{\capacity}{{\mathsf{cap}}}
\newcommand{\period}{{\mathsf{p}}}
\newcommand{\terminal}{\tau}

\newcommand{\parity}{{\mathsf{b}}}
\newcommand{\ratio}{\rho}
\newcommand{\Stether}{{\mathcal{C}}}
\newcommand{\Array}{{\mathcal{B}}}
\newcommand{\Int}[1]{{\left[{#1}\right\rangle}}
\newcommand{\resp}{respectively}
\newcommand{\Title}{On Bi-Modal Constrained Coding}
\newcommand{\Namea}{Ron M. Roth}
\newcommand{\Nameb}{Paul H. Siegel}
\newcommand{\Addressa}{Computer Science Department}
\newcommand{\Addressatwo}{Technion, Haifa 320003, Israel}
\newcommand{\Addressb}{ECE Department and CMRR}
\newcommand{\Addressbtwo}{UC San Diego, La Jolla, CA 92023, USA}

\newcommand{\Emaila}{ronny@cs.technion.ac.il}
\newcommand{\Emailb}{psiegel@ucsd.edu}
\newcommand{\Grant}{This work was supported by 
            Grants~2015816 and~2018048 from
            the United-States--Israel Binational
            Science Foundation (BSF), by NSF Grant CCF-BSF-1619053,
            and by Grant~1396/16 from the Israel Science Foundation.}
\newcommand{\Addressalt}{This work was done in part while R.M. Roth
            was visiting the Center for Memory and Recording
            Research (CMRR), UC San Diego}
\newcommand{\Conference}{
            This work, under the title
            ``On parity-preserving constrained coding,''
            was presented in part at the
            IEEE Int'l Symposium on Information Theory, 
            Vail, Colorado (June 2018).} 
\newcommand{\Thnxa}{\Namea\ is with the \Addressa, \Addressatwo.
                    \Addressalt.
                    Email: \Emaila}
\newcommand{\Thnxb}{\Nameb\ is with the \Addressb, \Addressbtwo.
                    Email: \Emailb}

\begin{document}
\ifIEEE
   \title{\Title}
          \author{\Namea
                  \quad\quad
                  \Nameb
                  \thanks{\Grant\Conference}
                  \thanks{\Thnxa}
                  \thanks{\Thnxb}}
\else
   \title{\textbf{\Title}\thanks{\Grant\Conference}}
   \author{\textsc{\Namea}\thanks{\Thnxa}
   \and
          \textsc{\Nameb}\thanks{\Thnxb}
   }
\fi
\maketitle

\begin{abstract}
Bi-modal (\resp, multi-modal) constrained coding refers
to an encoding model whereby
a user input block can be mapped to two (\resp, multiple)
codewords.  In current storage applications, such as optical disks,
multi-modal coding allows to achieve DC control, in addition to
satisfying the runlength limited (RLL) constraint specified
by the recording channel. In this work, a study is initiated on
bi-modal fixed-length constrained encoders.
Necessary and sufficient conditions are presented for
the existence of such encoders for a given constraint.
It is also shown that under somewhat stronger conditions, 
one can guarantee a bi-modal encoder with finite decoding delay.
\end{abstract}


\section{Introduction}
\label{sec:introduction}

Runlength limited (RLL) coding is widely employed in magnetic and
optical storage in order to mitigate the effects of inter-symbol
interference and clock drifting~\cite{Immink}. The encoder typically
takes the form of a finite-state machine, which maps a sequence
of input $p$-bit blocks into a sequence of output $q$-bit codewords,
so that the concatenation of the generated codewords satisfies
the RLL constraint. In most applications, the coding scheme also
provides DC control (or, more generally, suppression of
the low frequencies). This is achieved by 
\emph{bi-modal} (\resp, \emph{multi-modal})
\emph{encoding}, allowing some (or all)
input $p$-blocks to be mapped by the encoder
to two (\resp, multiple) codewords,
and then, during the encoding process, selecting
the codeword that yields the best DC suppression.
In one implementation of this strategy, the input sequence
is sub-divided into non-overlapping windows
(each consisting of one or more $p$-blocks),
and each window can be mapped by
the encoder to two output sequences over $\{ 0, 1 \}$
that have different parities (i.e., modulo-$2$ sums).
The generated output binary sequence is then transformed
(``precoded'') into sequences over the bipolar alphabet
$\{ +1, -1 \}$, with the binary $1$s corresponding to
positions of the bipolar sign changes~\cite[p.~52]{Immink}.
DC control is achieved by selecting the output sequence
that minimizes the DC contents~\cite[p.~29]{MRS}. 
A \emph{parity-preserving} encoder is one embodiment of this
approach and is part of the Blu-ray standard.
In such an encoder,
the parity of the input sub-sequence within each window
is preserved at the generated output sequence;
one bit in each window is then reserved to set the parity
of (the input window and) the output sequence, thereby
controlling the DC contents~\cite[\S 11.4.3]{Immink},
\cite{KI},
\cite{MSIWUYN},
\cite{NKHSEKDW},
\cite{NY},
\cite{WIXC}.
Another implementation of DC control via multi-modal encoding 
maps the input sequence (without precoding) into bipolar sequences
with different \emph{disparities}
(i.e., different signs of their sums of entries).
Multi-modal encoding can be
combined with guided scrambling~\cite[Ch.~10]{Immink}.

Most constructions so far of multi-modal encoders and, in particular,
of parity-preserving encoders, were obtained
by ad-hoc methods. The purpose of this work is to initiate a study
of bi-modal (and multi-modal) encoders, starting with the special case 
of fixed-length finite-state encoders.
We provide a formal definition of our setting
in Section~\ref{sec:ppencoders} below,
following a summary (in Section~\ref{sec:background})
of some background and definitions which are taken from~\cite{MRS}.
The main result of the paper, which we state in
Section~\ref{sec:mainresult}, is a necessary and sufficient
condition for the existence of a bi-modal encoder.
We prove the necessity part in Section~\ref{sec:necessity},
followed by a construction method of bi-modal encoders
in Section~\ref{sec:sufficiency}, thereby establishing sufficiency.
The more general variable-length model is a subject of
future work~\cite{RS} and is briefly discussed
in Section~\ref{sec:vle}.

\subsection{Background}
\label{sec:background}

A (finite labeled directed) graph is a graph $G = (V,E,L)$
with a nonempty finite state (vertex) set $V = V(G)$,
finite edge set $E = E(G)$, and edge labeling
$L : E \rightarrow \Sigma$.
The constraint presented by~$G$, denoted $S(G)$,
is the set of words over~$\Sigma$ that are generated
by finite paths in~$G$.
The set of words that are generated by finite paths that start
at a given state $u \in V(G)$ is called
the follower set of~$u$ and is denoted~$\FF_G(u)$.

A graph~$G$ is deterministic if
all outgoing edges from a state are distinctly labeled.
Every constraint has a deterministic presentation.
A graph~$G$ is lossless if no two paths with the same initial state
and the same terminal state generate the same word.
The anticipation~$\anticipation(G)$ of~$G$
is the smallest nonnegative integer~$a$
(if any) such that all paths that generate any given
word of length~$a{+}1$ from any given state in~$G$
share the same first edge (thus, the anticipation of
deterministic graphs is~$0$; generally, having finite anticipation is
a stronger property than losslessness).
A graph is $(m,a)$-definite if all paths that generate
a given word of length $m{+}a{+}1$ share the same $(m{+}1)$st edge.
A graph is said to have finite memory~$\mu$ if~$\mu$ is the smallest
nonnegative integer (if any) such that all paths of length~$\mu$
that generate the same word terminate in the same state.

A graph~$G$ is irreducible if it is strongly-connected.
Every graph decomposes uniquely into irreducible components,
and at least one such component must be an irreducible sink
(i.e., with no outgoing edges to another component).
The period of an irreducible graph (with at least one edge)
is the greatest common divisor of the lengths of its cycles,
and such a graph is primitive if its period is~$1$.
A constraint~$S$ is irreducible
if it can be presented by a deterministic irreducible graph
(in which case the smallest such presentation is unique).

The power~$G^t$ of a graph~$G$ is
the graph with the same set of states~$V(G)$
and edges that are the paths of length~$t$ in~$G$; the label 
of an edge in~$G^t$ is the length-$t$ word generated by the path.
For $S = S(G)$ the power~$S^t$ is defined as~$S(G^t)$.

Given a constraint~$S$ over an alphabet~$\Sigma$
and a lossless presentation~$G$ of~$S$,
the capacity of~$S$ is defined by
$\capacity(S) =
\lim_{\ell \rightarrow \infty} (1/\ell) \log_2 |S \cap \Sigma^\ell|$.
The limit indeed exists, and it is known that
$\capacity(S) = \log_2 \lambda(A_G)$ where~$\lambda(A_G)$ denotes
the spectral radius (Perron eigenvalue) of the adjacency matrix~$A_G$.

Given a constraint~$S$ and a
nonnegative\footnote{%
In all practical cases~$n$ needs to be strictly positive.
Yet for the purpose of simplifying the wording of some results in
the sequel, we find it convenient to allow~$n$ formally to be zero.}
integer~$n$, an $(S,n)$-encoder is a lossless graph~$\encoder$
such that $S(\encoder) \subseteq S$ and
each state has out-degree~$n$.
An $(S,n)$-encoder exists
if and only if~$\log_2 n \le \capacity(S)$.
In a \emph{tagged} $(S,n)$-encoder,
each edge is assigned an input tag from a finite alphabet of size~$n$,
such that edges outgoing from the same state have distinct tags.
The anticipation (if finite) of an encoder determines its decoding
delay: given the current encoder state, the anticipation specifies
how far one needs to look-ahead at a sequence generated from that state
in order to recover the first edge (and, thus, the tag of that edge)
in the encoder path that generated that sequence.
A tagged encoder is
$(m,a)$-sliding-block decodable if all paths that generate
a given word of length $m{+}a{+}1$ share the same tag
on their $(m{+}1)$st edges (thus, an $(m,a)$-definite encoder is
$(m,a)$-sliding-block decodable for any tagging of its edges).

A (tagged) rate $p:q$ encoder for a constraint~$S$ is
a (tagged) $(S^q,2^p)$-encoder (the tags are then assumed to be from
$\{ 0, 1 \}^p$). A rate $p:q$ parity-preserving encoder for
a constraint~$S$ over $\Sigma = \{ 0, 1 \}$ is 
a tagged encoder for~$S$ in which
the parity of the (length-$q$) label of each edge matches
the parity of the (length-$p$) tag that is assigned to
the edge (see also Section~\ref{sec:ppencoders} below).

Given a square nonnegative integer matrix~$A$ and a positive
integer~$n$, an $(A,n)$-approximate eigenvector
is a nonnegative nonzero integer vector~$\bldx$ that satisfies
the inequality $A \bldx \ge n \bldx$ componentwise.
The set of all $(A,n)$-approximate eigenvectors will be denoted by
$\XX(A,n)$, and it is known that
$\XX(A,n) \ne \emptyset$ if and only if $n \le \lambda(A)$.
Given a constraint~$S$ presented by a deterministic graph~$G$ and
a positive integer~$n$, the state-splitting algorithm provides
a method for transforming $G$, through
an $(A_G,n)$-approximate eigenvector, into
an $(S,n)$-encoder with finite anticipation.

For a positive integer~$b$,
the set $\{ 0, 1, 2, \ldots, b{-}1 \}$
will be denoted by~$\Int{b}$.

\subsection{Bi-modal encoders and parity-preserving encoders}
\label{sec:ppencoders}

Let $S$ be a constraint over an alphabet~$\Sigma$, and fix
a partition $\{ \Sigma_0, \Sigma_1 \}$ of $\Sigma$.
The symbols in~$\Sigma_0$ (\resp, $\Sigma_1$)
will be referred to as
the \emph{even} (\resp, \emph{odd}) symbols of $\Sigma$.
The case where one of the partition elements is empty will
turn out to be uninterestingly trivial, so we assume that~$\Sigma_0$
and~$\Sigma_1$ are both nonempty.
The partition of~$\Sigma$ to two elements (only)
and the choice of the terms even and odd follow from
the above-mentioned parity-based DC control approach
of constructing bi-modal encoders where an input 
$p$-bit block is mapped into two $q$-bit codewords
with different parities.
However, without much further effort, the definitions
and the results can be extended to
the multi-modal case, where~$\Sigma$ is partitioned into any
number of partition elements.\footnote{\label{footnote:disparity}%
Moreover, we will also consider
in Section~\ref{sec:non-partition}
the extension where~$\Sigma_0$ and~$\Sigma_1$ are not necessarily
disjoint (yet still cover all the elements of~$\Sigma$),
as this extension fits better the disparity-based DC control scheme:
in this case,
$\Sigma_0$ and~$\Sigma_1$ represent sets of bipolar $q$-blocks
of non-negative and non-positive disparities, respectively
and, so, they may intersect on $q$-blocks with zero disparity.}

Given a graph~$H$ with labeling in~$\Sigma$,
for $\parity \in \Int{2}$,
we denote by~$H_\parity$ the subgraph of~$H$ containing only the edges
with labels in~$\Sigma_\parity$.

Let $S = S(G)$ be a constraint and~$n_0$ and~$n_1$ be nonnegative
integers. An \emph{$(S,n_0,n_1)$-encoder} $\encoder$ is
an $(S,n_0{+}n_1)$-encoder such that for each~$\parity \in \Int{2}$,
the subgraph~$\encoder_\parity$ is an~$(S,n_\parity)$-encoder.
In other words, from each state in~$\encoder$,
there are~$n_0$ outgoing edges with even labels and~$n_1$ outgoing edges
with odd labels. Thus, in the parity-based DC control scheme,
a rate $p:q$ bi-modal encoder
for a constraint~$S$ over the binary alphabet
is an $(S^q,2^{p-1},2^{p-1})$-encoder,
where we partition the alphabet $S \cap \{ 0, 1 \}^q$ of $S^q$
into even and odd symbols according to their ordinary parity.
In particular, a rate $p:q$ parity-preserving
encoder for~$S$ is a (tagged) $(S^q,2^{p-1},2^{p-1})$-encoder and,
conversely, any $(S^q,2^{p-1},2^{p-1})$-encoder can be tagged so that
it is parity preserving. In the disparity-based DC control scheme,
the constraint~$S$ is over the bipolar alphabet
and we partition the set of length-$q$ words of~$S$
according to the sign of their disparity 
(zero-disparity words can be treated arbitrarily,
but see Footnote~\ref{footnote:disparity}).
In general, by a rate $p:q$ bi-modal encoder for a constraint~$S$
we hereafter mean
an $(S^q,2^{p-1},2^{p-1})$-encoder with respect to a prescribed
partition of the set of length-$q$ words of~$S$.

\begin{example}
\label{ex:paritypreserving}
Let~$S$ be the constraint over
$\Sigma = \{ a, b, c, d \}$ 
which is presented by the graph~$\encoder$ in
Figure~\ref{fig:paritypreserving}.
The graph~$\encoder$ is actually a deterministic
$(S,4)$-encoder.
Moreover, it is an $(S,2,2)$ encoder
if we assume the partition $\Sigma_0 = \{ a, b \}$ and
$\Sigma_1 = \{ c, d \}$:
the $(S,2)$-encoders $\encoder_0$ and $\encoder_1$ are
shown in Figures~\ref{fig:paritypreserving0}
and~\ref{fig:paritypreserving1}, respectively.\qed
\begin{figure}[t!]
\begin{center}
\thicklines
\setlength{\unitlength}{0.48mm}
\begin{picture}(150,60)(-35,-20)
    \multiput(000,000)(060,000){2}{\circle{20}}
    \qbezier(051,4.5)(030,14.5)(009,4.5)
    \put(051,4.5){\vector(2,-1){0}}
    \qbezier(054,8.3)(030,28.8)(006,8.3)
    \put(054,8.3){\vector(5,-4){0}}
    \qbezier(058,9.9)(030,48.0)(002,9.9)
    \put(058,9.9){\vector(2,-3){0}}
    \qbezier(051,-4.5)(030,-14.5)(009,-4.5)
    \put(009,-4.5){\vector(-2,1){0}}

    \qbezier(-9.26,004)(-30,015)(-30,000)
    \qbezier(-9.26,-04)(-30,-15)(-30,000)
    \put(-9.26,004){\vector(2,-1){0}}

    \qbezier(69.26,004)(090,015)(090,000)
    \qbezier(69.26,-04)(090,-15)(090,000)
    \put(69.26,004){\vector(-2,-1){0}}
    \qbezier(67.14,007)(100,025)(100,000)
    \qbezier(67.14,-07)(100,-25)(100,000)
    \put(67.14,007){\vector(-3,-2){0}}
    \qbezier(64.4,009)(110,035)(110,000)
    \qbezier(64.4,-09)(110,-35)(110,000)
    \put(64.4,009){\vector(-3,-2){0}}

\put(000,000){\makebox(0,0){$\alpha$}}
\put(060,000){\makebox(0,0){$\beta$}}

\put(-32,000){\makebox(0,0)[r]{$a$}}
\put(030,031.0){\makebox(0,0)[b]{$b$}}
\put(030,020.8){\makebox(0,0)[b]{$c$}}
\put(030,011.0){\makebox(0,0)[b]{$d$}}
\put(030,-11.5){\makebox(0,0)[t]{$d$}}
\put(092,000){\makebox(0,0)[l]{$a$}}
\put(102,001){\makebox(0,0)[l]{$b$}}
\put(112,000){\makebox(0,0)[l]{$c$}}
\end{picture}
\thinlines
\setlength{\unitlength}{1pt}
\end{center}
\caption{Graph~$\encoder$ for
Example~\protect\ref{ex:paritypreserving}.}
\label{fig:paritypreserving}
\end{figure}
\begin{figure}[h!t]
\begin{center}
\thicklines
\setlength{\unitlength}{0.48mm}
\begin{picture}(150,32)(-35,-16)
    \multiput(000,000)(060,000){2}{\circle{20}}
    \put(010,000){\vector(1,0){40}}

    \qbezier(-9.26,004)(-30,015)(-30,000)
    \qbezier(-9.26,-04)(-30,-15)(-30,000)
    \put(-9.26,004){\vector(2,-1){0}}

    \qbezier(69.26,004)(090,015)(090,000)
    \qbezier(69.26,-04)(090,-15)(090,000)
    \put(69.26,004){\vector(-2,-1){0}}
    \qbezier(67.14,007)(100,025)(100,000)
    \qbezier(67.14,-07)(100,-25)(100,000)
    \put(67.14,007){\vector(-3,-2){0}}

\put(000,000){\makebox(0,0){$\alpha$}}
\put(060,000){\makebox(0,0){$\beta$}}

\put(-32,000){\makebox(0,0)[r]{$a$}}
\put(030,002){\makebox(0,0)[b]{$b$}}
\put(092,000){\makebox(0,0)[l]{$a$}}
\put(102,001){\makebox(0,0)[l]{$b$}}
\end{picture}
\thinlines
\setlength{\unitlength}{1pt}
\end{center}
\caption{Graph~$\encoder_0$ for
Example~\protect\ref{ex:paritypreserving}.}
\label{fig:paritypreserving0}
\end{figure}
\begin{figure}[h!t]
\begin{center}
\thicklines
\setlength{\unitlength}{0.48mm}
\begin{picture}(150,45)(-35,-15)
    \multiput(000,000)(060,000){2}{\circle{20}}
    \qbezier(051,4.5)(030,14.5)(009,4.5)
    \put(051,4.5){\vector(2,-1){0}}
    \qbezier(054,8.3)(030,28.8)(006,8.3)
    \put(054,8.3){\vector(5,-4){0}}
    \qbezier(051,-4.5)(030,-14.5)(009,-4.5)
    \put(009,-4.5){\vector(-2,1){0}}

    \qbezier(69.26,004)(090,015)(090,000)
    \qbezier(69.26,-04)(090,-15)(090,000)
    \put(69.26,004){\vector(-2,-1){0}}

\put(000,000){\makebox(0,0){$\alpha$}}
\put(060,000){\makebox(0,0){$\beta$}}

\put(030,020.8){\makebox(0,0)[b]{$c$}}
\put(030,011.0){\makebox(0,0)[b]{$d$}}
\put(030,-11.5){\makebox(0,0)[t]{$d$}}
\put(092,000){\makebox(0,0)[l]{$c$}}
\end{picture}
\thinlines
\setlength{\unitlength}{1pt}
\end{center}
\caption{Graph~$\encoder_1$ for
Example~\protect\ref{ex:paritypreserving}.}
\label{fig:paritypreserving1}
\end{figure}
\end{example}

When studying $(S(G),n_0,n_1)$-encoders, there is no loss of
generality in assuming that both~$G$ and the encoder are
irreducible. Indeed, if~$\encoder$ is an $(S(G),n_0,n_1)$-encoder,
then an irreducible sink of~$\encoder$ is
an $(S',n_0,n_1)$-encoder, where~$S'$ is an irreducible constraint
presented by some irreducible component of~$G$
(see~\cite[Lemma~2.9]{MRS}).
Note that~$G_0$, $G_1$, $\encoder_0$, and~$\encoder_1$ may still
be reducible even when~$G$ and~$\encoder$ are irreducible
(e.g., in Example~\ref{ex:paritypreserving},
the graph~$\encoder$ is irreducible, while~$\encoder_0$ is not).

\subsection{Statement of main result}
\label{sec:mainresult}

Next is a statement of our main result.

\begin{theorem}
\label{thm:main}
Let $S$ be an irreducible constraint, presented by
an irreducible deterministic graph $G$, and let~$n_0$ and~$n_1$ be
positive integers.
Then there exists an $(S,n_0,n_1)$-encoder if and only if
$\XX(A_{G_0},n_0) \cap \XX(A_{G_1},n_1) \ne \emptyset$.
\end{theorem}

We prove Theorem~\ref{thm:main} in Sections~\ref{sec:necessity}
(necessity) and~\ref{sec:sufficiency} (sufficiency).
Hereafter, we use the notation
$\XX(A_{G_0},A_{G_1},n_0,n_1)$ for the intersection
$\XX(A_{G_0},n_0) \cap \XX(A_{G_1},n_1)$.

In view of Theorem~\ref{thm:main}, finding
the possible pairs~$(n_0,n_1)$
for which an~$(S(G),n_0,n_1)$-encoder exists
for a given~$G$ and partition~$\{ \Sigma_0, \Sigma_1 \}$ requires
a method for deciding whether
$\XX(A_{G_0},n_0)$ and $\XX(A_{G_1},n_1)$ share common vectors.
This decision problem can be recast
as a linear programming problem, namely, deciding whether
there is a \emph{real} vector~$\bldx$ that satisfies
the following constraints:
\begin{equation}
\label{eq:LP}
\begin{array}{rcl}
(A_{G_0} - n_0 I) \, \bldx & \ge & \bldzero \\
(A_{G_1} - n_1 I) \, \bldx & \ge & \bldzero \\
\bldx & \ge & \bldzero \\
\bldone^\top \cdot \bldx & = & 1 \; ,
\end{array}
\end{equation}
where~$\bldzero$ and~$\bldone$ stand for
the all-zero and all-$1$ column vectors
and $(\cdot)^\top$ denotes transposition.
Since all the coefficients in~(\ref{eq:LP}) are integers,
if there is a real
feasible solution~$\bldx$ then there is also a rational solution,
and, therefore, there is a nonzero integer solution
that satisfies the (first) three inequalities in~(\ref{eq:LP}).
There are known
polynomial-time algorithms for solving linear programming
problems, such as Karmarkar's algorithm~\cite{LY},
but it would be interesting to find a more direct method,
tailored specifically to the constraints~(\ref{eq:LP}),
for determining whether
$\XX(A_{G_0},A_{G_1},n_0,n_1)$ is nonempty
(see also the modified Franaszek algorithm in
Figure~\ref{fig:fz} in the sequel). 
In comparison, recall that in
the context of ordinary $(S,n)$-encoders,
the question of interest is whether $\XX(A_G,n)$ is nonempty,
which, in turn, is equivalent to asking whether $n \le \lambda(A_G)$.

\subsection{Going to powers of the constraint}
\label{sec:powers}

Next, we discuss the effect of going to powers of a constraint,
namely, attempting to construct $(S^t, n_0,n_1)$-encoders,
for increasing values of~$t$. To this end, we first need to define
the even and odd symbols in $\Sigma^t$, which is the alphabet\footnote{%
The \emph{effective} alphabet of $S^t$ is the subset
$S \cap \Sigma^t$ of $\Sigma^t$.}
of $S^t$, given a partition $\{ \Sigma_0, \Sigma_1 \}$ of~$\Sigma$.
We adopt a definition that fits the parity-based scheme for DC control:
we say that $\bldw \in \Sigma^t$ is even
(\resp, odd), if it contains an even (\resp, odd)
number of symbols from~$\Sigma_1$ (i.e., the parity
of~$\bldw$ is the modulo-$2$ sum of the parities of
the symbols in~$\bldw$).
The set of even (\resp, odd) words in~$\Sigma^t$ will be
denoted by~$(\Sigma^t)_0$ (\resp, $(\Sigma^t)_1$).

It turns out that in most cases,
we can approach the capacity of~$S$ with bi-modal encoders
if we let~$t$ increase. Note, however, that such an increase
may sometimes be \emph{necessary},
even when the capacity of $S$ is $\log_2 (n_0{+}n_1)$
(see Example~\ref{ex:twostates} below).
This presents a distinction between
bi-modal encoders and ordinary ones:
when $\capacity(S) = \log_2 n$, capacity is always attained
by ordinary encoders already for~$t = 1$.

Specifically, we have the following result.

\begin{theorem}
\label{thm:primitive}
Let~$G$ be a deterministic primitive graph,
having at least one edge with an even label
and one edge with an odd label.
Then there exists an infinite sequence of nonnegative integers
$n^{(1)}, n^{(2)}, \cdots$
such that $(S(G^t),n^{(t)},n^{(t)})$-encoders exist and
\[
\lim_{t \rightarrow \infty} (\log_2 n^{(t)})/t = \capacity(S(G)) \; .
\]
\end{theorem}

\begin{proof}
By assumption, there are edges
$e_0 = v_0 \stackrel{a}{\rightarrow} v'_0$
and
$e_1 = v_1 \stackrel{b}{\rightarrow} v'_1$
in~$G$ such that~$a$ is even and~$b$ is odd
($v_0$ and~$v_1$ may be the same state, and so may~$v'_0$ and~$v'_1$).
Also, since~$G$ is primitive, there exists a real $\kappa > 0$
such that $(A_G^\ell)_{u,v} > \kappa \cdot \lambda^\ell$
for any sufficiently large~$\ell$ and for every $u, v \in V(G)$,
where $\lambda = \lambda(A_G)$ (see~\cite[\S 3.3.4]{MRS}).

Fix~$u^*$ to be some $G$-state. We show that there are two
paths in~$G$ of the same length~$\ell^*$ that start at~$u^*$
and have different parities. Let~$\ell$ 
be sufficiently large such that $A_G^\ell > 0$ (componentwise),
and consider all paths of length~$\ell$ starting at~$u^*$.
If two of them have different parities, we are done,
with $\ell^*$ taken as~$\ell$.
Otherwise, for $\parity \in \Int{2}$, let $\pi_\parity$ be a path
of length~$\ell$ in~$G$ that starts at~$u^*$ and
terminates in~$v_\parity$.
The paths $\pi_0 e_0$ and~$\pi_1 e_1$, of length~$\ell^* = \ell{+}1$,
then have different parities.

Now, if $t{-}\ell^*$ is sufficiently large, there are at least
$n^{(t)} = \lceil \kappa \cdot \lambda^{t-\ell^*} \rceil$
paths of length~$t{-}\ell^*$ from any state~$u \in V(G)$ to~$u^*$.
It follows that from any state~$u \in V(G)$, there are at least
$n^{(t)}$ paths of length~$t$ that generate even words,
and at least $n^{(t)}$ paths of that length that generate odd words.
This means that~$G^t$ contains a subgraph which is
an $(S(G^t),n^{(t)},n^{(t)})$-encoder.
The result follows by noticing that
$\lim_{t \rightarrow \infty} (\log_2 n^{(t)})/t = \log_2 \lambda =
\capacity(S(G))$.
\end{proof}

\begin{remark}
\label{rem:primitive}
The requirement in Theorem~\ref{thm:primitive}
of having both even and odd labels
in~$G$ is necessary. Indeed, if all the labels in~$G$ were even,
then so would have to be all the labels in the encoder, implying
that $n^{(t)} = 0$. Similarly, if all the labels in~$G$ were odd,
then all the labels in the encoder would have to be even
(\resp, odd) for even (\resp, odd) $t$,
again implying that $n^{(t)} = 0$.\qed
\end{remark}

For a deterministic graph~$G$ and a positive integer~$t$,
we denote by $n_{\max}(G,t)$ the largest integer~$n$ for which there
exist $(S(G^t),n,n)$-encoders, and define the largest possible
coding ratio attainable by such encoders by
\begin{equation}
\label{eq:codingratio}
\ratio(G,t) = 
\frac{\log_2 \, (2 \, n_{\max}(G,t))}{t}
\; .
\end{equation}

\begin{example}
\label{ex:twostates}
Let~$G$ be the graph in Figure~\ref{fig:twostates}, where
\begin{equation}
\label{eq:partitioning}
\Sigma_0 = \{ a, b \} \quad \textrm{and} \quad
\Sigma_1 = \{ c, d \} \; .
\end{equation}
\begin{figure}[hbt]
\begin{center}
\thicklines
\setlength{\unitlength}{0.48mm}
\begin{picture}(105,45)(-35,-15)
    \multiput(000,000)(060,000){2}{\circle{20}}
    \qbezier(051,4.5)(030,14.5)(009,4.5)
    \put(051,4.5){\vector(2,-1){0}}
    \qbezier(054,8.3)(030,28.8)(006,8.3)
    \put(054,8.3){\vector(5,-4){0}}
    \qbezier(051,-4.5)(030,-14.5)(009,-4.5)
    \put(009,-4.5){\vector(-2,1){0}}

    \qbezier(-9.26,004)(-30,015)(-30,000)
    \qbezier(-9.26,-04)(-30,-15)(-30,000)
    \put(-9.26,004){\vector(2,-1){0}}

\put(000,000){\makebox(0,0){$\alpha$}}
\put(060,000){\makebox(0,0){$\beta$}}

\put(-32,000){\makebox(0,0)[r]{$a$}}
\put(030,020.8){\makebox(0,0)[b]{$b$}}
\put(030,011.0){\makebox(0,0)[b]{$c$}}
\put(030,-11.5){\makebox(0,0)[t]{$d$}}
\end{picture}
\thinlines
\setlength{\unitlength}{1pt}
\end{center}
\caption{Graph~$G$ for Example~\protect\ref{ex:twostates}.}
\label{fig:twostates}
\end{figure}
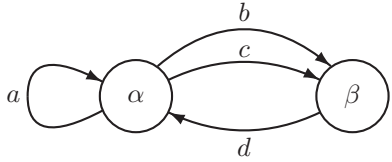
We have $\lambda(A_G) = 2$,
and the matrices $A_{G_0}$ and $A_{G_1}$ are given by
\[
A_{G_0} =
\left(
\begin{array}{cc}
1 & 1 \\
0 & 0
\end{array}
\right)
\quad \textrm{and} \quad
A_{G_1} =
\left(
\begin{array}{cc}
0 & 1 \\
1 & 0
\end{array}
\right)
\; .
\]
Letting $S = S(G)$,
we find conditions under which there exist
$(S^t,n,n)$-encoders, for various values of~$t$.
Starting with~$t = 1$, in this case, 
$\XX(A_{G_0},A_{G_1},n,n) = \emptyset$
for every $n > 0$. Indeed, when~$n > 0$,
any $\bldx = (x_\alpha \; x_\beta)^\top \in \XX(A_{G_0},n)$ 
must have $x_\beta = 0$ (since~$\beta$ has no outgoing edges
in~$G_0$),
but then $A_{G_1} \bldx \ge n \, \bldx$ implies that~$x_\alpha = 0$.
Hence, $n_{\max}(G,1) = 0$ and~$\ratio(G,1) = -\infty$.

Next, we turn to the second power of~$G$.
Here $S \cap (\Sigma^2)_0 = \{ aa, ab, cd, dc \}$,
$S \cap (\Sigma^2)_1 = \{ ac, bd, da, db \}$,
and, respectively,
\[
A_{(G^2)_0} =
\left(
\begin{array}{cc}
2 & 1 \\
0 & 1
\end{array}
\right)
\quad \textrm{and} \quad
A_{(G^2)_1} =
\left(
\begin{array}{cc}
1 & 1 \\
1 & 1
\end{array}
\right)
\; .
\]
We must have $n \le \lambda(A_{(G^2)_1}) = 2$.
Checking first the case $n = 2$, any
$\bldx \in \XX(A_{(G^2)_1},2)$
must be a multiple of $(1 \; 1)^\top$,
which is a (true) eigenvector of $A_{(G^2)_1}$
associated with the eigenvalue~$2$~\cite[Theorem~5.4]{MRS}.
Yet $(1 \; 1)^\top \not\in \XX(A_{(G^2)_0},2)$,
so we must have~$n \le 1$,
and it is easily seen that
$(1 \; 1)^\top \in \XX(A_{(G^2)_0},A_{(G^2)_1},1,1) \ne \emptyset$.
Therefore, by Theorem~\ref{thm:main} there exists
an $(S^2,1,1)$-encoder
(i.e., a rate $1:2$ bi-modal encoder for~$S$
with respect to the partition
$\left\{ (\Sigma^2)_0, (\Sigma^2)_1 \right\}$).
In this case, $n_{\max}(G,2) = 1$ and~$\ratio(G,2) = 1/2$.

Turning to the third power of~$G$, here
\[
\renewcommand{\arraystretch}{1.3}
\begin{array}{rcl}
S \cap (\Sigma^3)_0
& = & \{ aab, cda, cdb, dcd, acd, bdc, dac, dbd \} \; , \\
S \cap (\Sigma^3)_1
& = &\{ aac, abd, cdc, dcd, bda, bdb, daa, dab \} \; ,
\end{array}
\]
and, respectively,
\[
A_{(G^3)_0} =
\left(
\begin{array}{cc}
3 & 3 \\
1 & 1
\end{array}
\right)
\quad \textrm{and} \quad
A_{(G^3)_1} =
\left(
\begin{array}{cc}
2 & 3 \\
2 & 1
\end{array}
\right)
\; .
\]
We must have $n \le \lambda(A_{(G^3)_1}) = 4$,
and we again rule out the case $n = 4$ 
by observing that $\XX(A_{(G^3)_0},4)$
consists of multiples of $(3 \; 1)^\top$,
yet this vector is not in $\XX(A_{(G^3)_1},4)$.
Hence, we must have~$n \le 3$, and since
$(3 \; 1)^\top \in \XX(A_{(G^3)_0},A_{(G^3)_1},3,3)$
we get by Theorem~\ref{thm:main}
that there exists an $(S^3,3,3)$-encoder.
Thus, $n_{\max}(G,3) = 3$ and~$\ratio(G,3) = (\log_2 6)/3$.

In general, using the equalities
\[
A_{(G^t)_0} = A_{G_0} A_{(G^{t-1})_0}
+ A_{G_1} A_{(G^{t-1})_1} \; \phantom{,}
\]
and
\[
A_{(G^t)_1} = A_{G_0} A_{(G^{t-1})_1}
+ A_{G_1} A_{(G^{t-1})_0} \; ,
\]
it can be shown by induction on~$t$ that
\[
A_{(G^t)_0} =
\frac{1}{6}
\left(
\begin{array}{cc}
2^{t+1} + 3 + (-1)^t & 2^{t+1} - 2 \,(-1)^t \\
2^t     - 3 - (-1)^t & 2^t     + 2 \,(-1)^t
\end{array}
\right)
\; \phantom{.}
\]
and
\[
A_{(G^t)_1} =
\frac{1}{6}
\left(
\begin{array}{cc}
2^{t+1} - 3 + (-1)^t & 2^{t+1} - 2 \,(-1)^t \\
2^t     + 3 - (-1)^t & 2^t     + 2 \,(-1)^t
\end{array}
\right)
\; ,
\]
with $\lambda(A_{(G^t)_0}) = \lambda(A_{(G^t)_1}) = 2^{t-1}$
and respective eigenvectors
\[
\bldx_0 =
\left(
\begin{array}{c}
\arraycolsep0ex
2^{t+1} - 2 \, (-1)^t \\ 2^t - 3 - (-1)^t
\end{array}
\right)
\;\; \textrm{and} \;\;
\bldx_1 =
\left(
\arraycolsep0ex
\begin{array}{c}
2^{t+1} - 2 \, (-1)^t \\ 2^t + 3 - (-1)^t
\end{array}
\right)
\; .
\]
Clearly, each of the sets $\XX(A_{(G^t)_0},n)$ and $\XX(A_{(G^t)_1},n)$
is empty if and only if $n > 2^{t-1}$.
Their intersection, however,
is empty also when $n = 2^{t-1}$:
for $\parity \in \Int{2}$,
the elements of $\XX(A_{(G^t)_\parity},2^{t-1})$
are (true) eigenvectors of~$A_{(G^t)_\parity}$
associated with the Perron eigenvalue~$2^{t-1}$, namely,
they are scalar multiples of~$\bldx_\parity$,
yet~$\bldx_0$ and~$\bldx_1$ are linearly independent.
We therefore conclude that
$\XX(A_{(G^t)_0},A_{(G^t)_1},n_0,n_1) \ne \emptyset$ only when
$\max \{ n_0, n_1 \} \le 2^{t-1}$ and $n_0 + n_1 < 2^t$;
in particular, there can be no $(S^t,n_0,n_1)$-encoder when
$\log_2 (n_0 + n_1) = t$
(and this holds also when $n_0 \ne n_1$).
It follows that there can be
no rate $t:t$ bi-modal encoder for~$S$
(with respect to the partition
$\left\{ (\Sigma^t)_0, (\Sigma^t)_1 \right\}$),
for any positive integer~$t$.

On the other hand, for $n^{(t)} = 2^{t-1}{-}1$, we do have
$(2 \; 1)^\top\in \XX(A_{(G^t)_0},A_{(G^t)_1},n^{(t)},n^{(t)})$.
Therefore,
$\ratio(G,t) = (1/t) \log_2 (2^t{-}2) \; (< 1)$,
and
$\lim_{t \rightarrow \infty} \ratio(G,t) = 1 = \log_2 \lambda(A_G)$,
so the rate can \emph{approach} the capacity value of~$1$
as~$t$ increases (yet can never attain it).\qed
\end{example}

\begin{example}
\label{ex:twostatesalt}
Let~$G$ be as in Example~\ref{ex:twostates}, except that now
\begin{equation}
\label{eq:partitioningalt}
\Sigma_0 = \{ a \} \quad \textrm{and} \quad
\Sigma_1 = \{ b, c, d \} \; .
\end{equation}
It can be shown by induction on~$t$ that in this case,
\[
A_{(G^t)_0} =
\frac{1}{3}
\left(
\begin{array}{cc}
2^{t+1} + (-1)^t & 0 \\
0                & 2^t + 2\,(-1)^t
\end{array}
\right)
\; \phantom{.}
\]
and
\[
A_{(G^t)_1} =
\frac{1}{3}
\left(
\begin{array}{cc}
0            & 2^{t+1} - 2\,(-1)^t \\
2^t - (-1)^t & 0
\end{array}
\right)
\; .
\]
Letting
\[
n^{(t)} = (A_{(G^t)_0})_{\beta,\beta}
= (1/3) \left( 2^t + 2\,(-1)^t \right)
\]
and
\[
\lambda^{(t)} = \lambda \left( A_{(G^t)_1} \right)
= (\sqrt{2}/3) \left(2^t - (-1)^t \right) \; ,
\]
it can be verified that
$\XX(A_{(G^t)_0},A_{(G^t)_1},n_0,n_1) \ne \emptyset$ 
only when $n_0 \le n^{(t)}$ and $n_1 \le \lambda^{(t)}$.
In particular, since
$n^{(t)} + \lambda^{(t)} < 2^t$,
there can be no $(S^t,n_0,n_1)$-encoder when
$\log_2 (n_0 + n_1) = t$.

On the other hand, when $t > 2$, we do have
$(3 \; 2)^\top \in \XX(A_{(G^t)_0},A_{(G^t)_1},n^{(t)},n^{(t)})$.
Hence, $n_{\max}(G,t) = n^{(t)}$ and,
similarly to Example~\ref{ex:twostates},
$\ratio(G,t) = (1/t) \log_2 (2 n^{(t)}) \; (< 1)$:
we can approach the capacity value (of~$1$),
although we will never attain it.

Note that while the graphs~$G^t$ are primitive for all positive
integers~$t$, the graphs~$(G^t)_0$ in this example are reducible and
$(G^t)_1$ are irreducible with period~$2$.\qed
\end{example}

The result of Theorem~\ref{thm:primitive}
can be generalized to the non-primitive case as follows.
Let~$G$ be an irreducible deterministic graph with period~$\period$.
Then $G^\period$ decomposes into~$\period$ primitive
irreducible components~$G_0, G_1, \ldots, G_{\period-1}$,
with each~$G_i$ being the induced subgraph on 
a congruence class~$C_i$ in~$V(G)$. Suppose that for some~$i$,
there are two paths of length~$\period$ that start at states
in~$C_i$ and generate words with different parities.
Under these conditions, the subgraph~$G_i$ will satisfy
the conditions of Theorem~\ref{thm:primitive}.

Theorem~\ref{thm:primitive}
(along with the previous examples)
focused on $(S,n_0,n_1)$-encoders
where $n_0 = n_1$. While this case suits the motivation
of rate $p:q$ bi-modal encoders (where $n_0 = n_1 = 2^{p-1}$),
there seems to be a merit in studying the more general case as well
(see Example~\ref{ex:210rll} in Section~\ref{sec:anticipationone}).

\begin{example}
\label{ex:rateregion}
Let~$G$ be a deterministic graph with the adjacency matrix
\[
A_G =
\left(
\begin{array}{cc}
3 & 3 \\
5 & 5
\end{array}
\right)
\; ,
\]
and let~$\Sigma_0$ and~$\Sigma_1$ be such that
\[
A_{G_0} =
\left(
\begin{array}{cc}
1 & 2 \\
2 & 0
\end{array}
\right)
\quad \textrm{and} \quad
A_{G_1} =
\left(
\begin{array}{cc}
2 & 1 \\
3 & 5
\end{array}
\right)
\; .
\]
Then
\[
A_{(G^2)_0} =
\left(
\begin{array}{cc}
12 & 9 \\
23 & 32
\end{array}
\right)
\quad \textrm{and} \quad
A_{(G^2)_1} =
\left(
\begin{array}{cc}
12 & 15 \\
17 & 8
\end{array}
\right)
\; .
\]
It can be verified that $\lambda(A_{G^2}) = 64$,
$\lambda(A_{(G^2)_0}) \approx 39.5$, and
$\lambda(A_{(G^2)_1}) \approx 26.1$.
Figure~\ref{fig:rateregion} shows the boundary of
the region of all pairs~$(n_0,n_1)$
for which $\XX(A_{(G^2)_0},A_{(G^2)_1},n_0,n_1) \ne \emptyset$.\qed
\newcommand{\Bullet}{\circle*{5}}
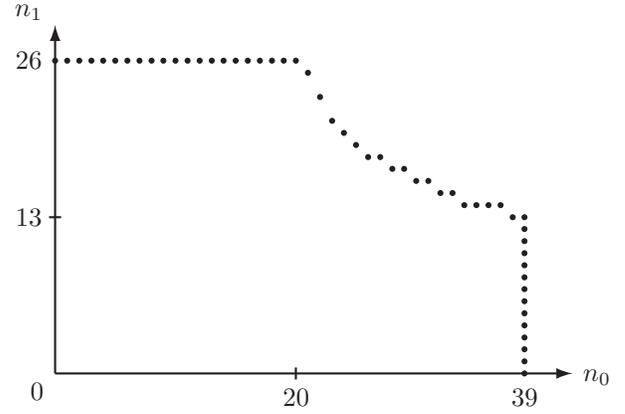
\begin{figure}[hbt]
\begin{center}
\thicklines
\setlength{\unitlength}{0.16mm}
\begin{picture}(480,330)(-25,-25)
\put(000,000){\vector(1,0){430}}
\put(000,000){\vector(0,1){290}}
\put(-15,-15){\makebox(0,0){$0$}}
\put(450,-03){\makebox(0,0){$n_0$}}
\put(-22,300){\makebox(0,0){$n_1$}}
\put(390,-20){\makebox(0,0){$39$}}
\put(200,-05){\line(0,1){10}}
\put(200,-20){\makebox(0,0){$20$}}
\put(-22,260){\makebox(0,0){$26$}}
\put(-05,130){\line(1,0){10}}
\put(-22,130){\makebox(0,0){$13$}}
\multiput(000,260)(010,000){21}{\Bullet}
\multiput(390,000)(000,010){13}{\Bullet}
\put(210,250){\Bullet}
\put(220,230){\Bullet}
\put(230,210){\Bullet}
\put(240,200){\Bullet}
\put(250,190){\Bullet}
\put(260,180){\Bullet}
\put(270,180){\Bullet}
\put(280,170){\Bullet}
\put(290,170){\Bullet}
\put(300,160){\Bullet}
\put(310,160){\Bullet}
\put(320,150){\Bullet}
\put(330,150){\Bullet}
\put(340,140){\Bullet}
\put(350,140){\Bullet}
\put(360,140){\Bullet}
\put(370,140){\Bullet}
\put(380,130){\Bullet}
\put(390,130){\Bullet}
\end{picture}
\thinlines
\setlength{\unitlength}{1pt}
\end{center}
\caption{Attainable pairs~$(n_0,n_1)$
for Example~\protect\ref{ex:rateregion}.}
\label{fig:rateregion}
\end{figure}
\end{example}

\section{Necessary condition}
\label{sec:necessity}

In this section we prove the ``only if'' part in Theorem~\ref{thm:main},
recalled next.

\begin{theorem}
\label{thm:necessity}
Let $S$ be an irreducible constraint, presented by
an irreducible deterministic graph $G$, and let~$n_0$ and~$n_1$ be
positive integers.
Then, an $(S,n_0,n_1)$-encoder exists, only if
$\XX(A_{G_0},A_{G_1},n_0,n_1) \ne \emptyset$, namely only if
there exists a nonnegative nonzero integer vector~$\bldx$ such that
\begin{equation}
\label{eq:commonAE}
A_{G_0} \bldx \ge n_0 \bldx
\quad \textrm{and} \quad
A_{G_1} \bldx \ge n_1 \bldx
\; .
\end{equation}
\end{theorem}

The proof of the theorem is a refinement of the proof of
Theorem~3 in~\cite{MR} (or Theorem~7.2 in~\cite{MRS}),
where it was shown, \emph{inter alia}, that
the existence of an~$(S,n)$-encoder implies
the existence of an $(A_G,n)$-approximate eigenvector.
Since the existence of an $(S,n_0,n_1)$-encoder 
implies the existence of $(S(G_\parity),n_\parity)$-encoders
for $\parity \in \Int{2}$,
it also implies
that~$\XX(A_{G_0},n_0)$ and~$\XX(A_{G_1},n_1)$ are both nonempty sets
(and, so, $n_0 \le \lambda(A_{G_0})$
and $n_1 \le \lambda(A_{G_1})$).
Theorem~\ref{thm:necessity} states
that their \emph{intersection} must be nonempty too.

\begin{proof}[Proof of Theorem~\ref{thm:necessity}]
Suppose that an~$(S,n_0,n_1)$-encoder exists,
and fix~$\encoder$ to be such an encoder.
Next, we follow the three steps of the proof in~\cite{MR}:
step~(a) is identical to the respective step in that proof,
and so are the definitions of~$H'$, $\bldc$, and~$\bldx$
in the remaining steps.

\emph{(a)
Construct a deterministic graph $H$ which presents $S(\encoder)$.}
For any word $\bldw$ and state
$v \in V(\encoder)$, let $T_\encoder(\bldw,v)$ denote
the subset of states in $\encoder$ which are accessible from $v$
by paths in $\encoder$ which generate $\bldw$
(when $\bldw$ is the empty word
define $T_\encoder(\bldw,v) = \{ v \}$).
The states of $H$ are defined as the distinct
nonempty subsets $\{ T_\encoder(\bldw,v) \}_{\bldw,v}$ of $V(\encoder)$,
and we endow~$H$ with an edge
$Z \stackrel{a}{\rightarrow} Z'$,
if and only if there exists
a state $v \in V(\encoder)$ and a word $\bldw$ such that
$Z = T_\encoder(\bldw,v)$ and $Z' = T_\encoder(\bldw a, v)$.
The graph~$H$ is deterministic by construction,
and by~\cite[Lemma~2.1]{MRS} we have $S(H) = S(\encoder)$.
For the rest of the proof, fix~$H'$ to be any irreducible sink of~$H$.
 
\emph{(b)
Define a positive integer vector $\bldc$ such that
$A_{H_0'} \bldc = n_0 \bldc$ and $A_{H_1'} \bldc = n_1 \bldc$.}
Recalling that each state $Z \in V(H')$
is a subset of $V(\encoder)$, let $c_Z = |Z|$ denote the number
of $\encoder$-states in $Z$ and let $\bldc$ be the positive
integer vector defined by
$\bldc = (c_Z)_{Z \in V(H')}$.
We claim that
\begin{equation}
\label{eq:AsubH}
A_{H_0'} \bldc = n_0 \bldc
\quad \textrm{and} \quad
A_{H_1'} \bldc = n_1 \bldc \; .
\end{equation}

Fix a parity $\parity \in \Int{2}$, and consider a state $Z \in V(H')$.
Since $\encoder_\parity$ has out-degree $n_\parity$, 
the number of edges in~$\encoder_\parity$
outgoing from the subset of states $Z \subseteq V(\encoder)$
is $n_\parity |Z|$
(this is also the number of outgoing edges from~$Z$ in~$\encoder$
with labels from~$\Sigma_\parity$).
For $a \in \Sigma_\parity$, let $E_a$ denote
the set of $\encoder_\parity$-edges labeled $a$ outgoing from
the $\encoder$-states in $Z$,
and let $Z_a$ denote the set of terminal $\encoder$-states of
these edges;
note that $\{ E_a \}_{a \in \Sigma_\parity}$ forms a partition of
the set of $\encoder_\parity$-edges outgoing from $Z$.
If $Z_a \ne \emptyset$, then clearly
there is an edge
$Z \stackrel{a}{\rightarrow} Z_a$
in $H$ and, since $H'$ is a sink, this edge is also contained in $H'$,
and, therefore, in~$H_\parity'$.
We now claim that any $\encoder$-state $u \in Z_a$ is accessible
in $\encoder$ (and in~$\encoder_\parity$)
by exactly one edge labeled~$a$ that starts at $Z$;
otherwise if $Z = T_\encoder(\bldw,v)$,
the word $\bldw a$ could be generated in $\encoder$ by two
distinct paths which start at~$v$ and terminate in~$u$,
contradicting the losslessness of $\encoder$.
Hence, $|E_a| = |Z_a|$ and, so, the entry of $A_{H_\parity'} \bldc$
corresponding to the $H'$-state $Z$ satisfies
\begin{eqnarray*}
(A_{H_\parity'} \bldc)_Z & = &
\sum_{Y \in V(H')} (A_{H_\parity'})_{Z,Y} c_Y
= \sum_{Y \in V(H')} (A_{H_\parity'})_{Z,Y} |Y| \\
& = &
\sum_{a \in \Sigma_\parity} |Z_a|
= \sum_{a \in \Sigma_\parity} |E_a|
= n_\parity |Z|
= n_\parity c_Z \; ,
\end{eqnarray*}
thus proving~(\ref{eq:AsubH}).
 
\emph{(c)
Construct from $\bldc$ a nonnegative nonzero integer vector~$\bldx$
that satisfies~(\ref{eq:commonAE}).}
Recalling the definition and notation of a follower set from
the beginning of Section~\ref{sec:background},
let $\bldx = (x_u)_{u \in V(G)}$ be
the nonnegative integer vector 
whose entries are defined for every $u \in V(G)$ by
\begin{equation}
\label{eq:xvector}
x_u = \max \Bigl\{  c_Z \;:\; Z \in V(H')
\; \textrm{and} \;
\FF_{H'}(Z) \subseteq \FF_G(u) \Bigr\} \; .
\end{equation}
Also, denote by $Z(u)$ some particular $H'$-state $Z$ for which
the maximum is attained in~(\ref{eq:xvector});
in case the maximum is over an empty set,
define $x_u = 0$ and $Z(u) = \emptyset$.
Since $S(H') \subseteq S(G)$, we get from~\cite[Lemma~2.13]{MRS}
that each follower set of an~$H'$-state is contained in
a follower set of some $G$-state;
hence, $\bldx$ is necessarily nonzero.

Next, we show that $\bldx$ satisfies~(\ref{eq:commonAE}).
Fix a parity~$\parity \in \Int{2}$,
and let $u$ be a $G$-state; if $x_u = 0$ then, trivially,
$(A_{G_\parity} \bldx)_u \ge n_\parity x_u$;
so, we assume hereafter that $x_u \ne 0$.
Let $\Sigma_\parity(Z(u))$ denote the set of labels of edges
in $H_\parity'$ outgoing from $Z(u)$, and, for
$a \in \Sigma_\parity(Z(u))$, let $Z_a(u)$
be the terminal state in $H'$ (and in~$H_\parity'$)
for an edge labeled~$a$ outgoing from the $H'$-state $Z(u)$.
Since $\FF_{H'}(Z(u)) \subseteq \FF_G(u)$,
there exists an edge labeled~$a$ in~$G$ from~$u$ which terminates in
some $G$-state $u_a$.
Now, since $G$ and $H'$ are both deterministic, we have
$\FF_{H'} (Z_a(u)) \subseteq \FF_G (u_a)$ and, so, 
by~(\ref{eq:xvector}) we get that $x_{u_a} \ge c_{Z_a(u)}$.
Therefore,
\begin{eqnarray*}
(A_{G_\parity} \bldx)_u
& \ge &
\sum_{a \in \Sigma_\parity(Z(u))} x_{u_a}
\ge \sum_{a \in \Sigma_\parity(Z(u))} c_{Z_a(u)} \\
& = &
(A_{H_\parity'} \bldc)_{Z(u)}
\stackrel{(\ref{eq:AsubH})}{=}
n_\parity c_{Z(u)}
= n_\parity x_u \; ,
\end{eqnarray*}
namely, $A_{G_\parity} \bldx \ge n_\parity \bldx$.
\end{proof}

The following corollary parallels Corollary~1 in~\cite{MR}
(or Theorem~7.2 in~\cite{MRS}).

\begin{corollary}
\label{cor:lowerboundstates}
Let $S$ be an irreducible constraint, presented by
an irreducible deterministic graph $G$, and let~$n_0$
and~$n_1$ be positive integers.
Then, for any $(S,n_0,n_1)$-encoder $\encoder$,
\[
|V(\encoder)| \ge
\min_{\bldx \in \XX(A_{G_0},A_{G_1},n_0,n_1)}
\| \bldx \|_\infty \; ,
\]
where~$\| (x_u)_u \|_\infty = \max_u x_u$.
\end{corollary}

\begin{proof}
Given an $(S,n_0,n_1)$-encoder~$\encoder$, construct the vector
$\bldx \in \XX(A_{G_0},A_{G_1},n_0,n_1)$ as
in the proof of Theorem~\ref{thm:necessity}.
Then, by construction, each component of~$\bldx$ is a size
of a subset of $V(\encoder)$ and, hence, bounds it from below.
\end{proof}

The next corollary parallels Theorem~5 in~\cite{MR}
(or Theorem~7.15 in~\cite{MRS})
and is proved in the very same manner.

\begin{corollary}
\label{cor:lowerboundanticipation}
With $S$, $G$, $n_0$, $n_1$, and~$\encoder$ as in
Corollary~\ref{cor:lowerboundstates},
\[
\anticipation(\encoder) \ge
\log_n
\Bigl(
\min_{\bldx \in \XX(A_{G_0},A_{G_1},n_0,n_1)}
\| \bldx \|_\infty
\Bigr)
\; ,
\]
where~$n = \max \{ n_0, n_1 \}$.
\end{corollary}

The Franaszek algorithm is a known method for computing
approximate eigenvectors~\cite[\S 5.2.2]{MRS}.
Figure~\ref{fig:fz} presents a modification of it for
computing a vector in~$\XX(A_0,A_1,n_0,n_1)$,
where~$A_0$ and~$A_1$ are nonnegative integer $k \times k$ matrices
(a nonnegative integer $k$-vector~$\bldxi$ is provided
as an additional parameter to the algorithm).
The modified algorithm can be used to compute the lower bounds
of Corollaries~\ref{cor:lowerboundstates}
and~\ref{cor:lowerboundanticipation}, and will turn out to be useful
also when designing bi-modal encoders.

By slightly generalizing the proof of validity of
the (ordinary) Franaszek algorithm (see~\cite[\S 5.2.2]{MRS}) 
it follows that the algorithm
in Figure~\ref{fig:fz} returns the largest
(componentwise) vector $\bldx \in \XX(A_0,A_1,n_0,n_1)$
that satisfies $\bldx \le \bldxi$; if no such
vector exists, then the algorithm returns the all-zero vector.

\begin{figure}[hbt]
\noindent
\makebox[0in]{}\hrulefill\makebox[0in]{}
{\footnotesize
\begin{quote}
{\tt
$\bldy \leftarrow \bldxi$; \\
$\bldx \leftarrow \bldzero$; \\
while ($\bldx \ne \bldy$) \quad $\{$ \\ 
\hspace*{2em}
$\bldx \leftarrow \bldy$; \\
\hspace*{2em}
$\bldy \leftarrow \min
\left\{
\left\lfloor \frac{1}{n_0} A_0 \bldx \right\rfloor,
\left\lfloor \frac{1}{n_1} A_1 \bldx \right\rfloor,
\bldx
\right\}$; \\
\hspace*{6em}
$/*$ {\rm apply $\lfloor \cdot \rfloor$ and $\min \{ \cdot,
\cdot \}$ componentwise} $*/$ \\ 
$\}$ \\
return $\bldx$; }
\end{quote}\vspace{-2ex}
}
\noindent
\makebox[0in]{}\hrulefill\makebox[0in]{}
\caption{Modified Franaszek algorithm.}
\label{fig:fz}
\end{figure}

By analyzing the complexity of Karmarkar's algorithm~\cite{LY}
(in terms of number of bit operations),
one can conceptually infer an upper bound on the smallest possible
largest entry of any vector in $\XX(A_0,A_1,n_0,n_1)$, in terms of
$A_0$, $A_1$, $n_0$, and~$n_1$ (provided that we first use
that algorithm to determine that $\XX(A_0,A_1,n_0,n_1)$ is nonempty).
Equivalently, such an analysis implies an upper bound on
the smallest possible integer~$\xi > 0$ such that
running the algorithm in Figure~\ref{fig:fz}
with $\bldxi = \xi \cdot \bldone$ yields a nonzero output
$\bldx \in \XX(A_0,A_1,n_0,n_1)$.
It is still open whether there is a more direct way for computing
(efficiently) such an upper bound.

\section{Sufficient condition}
\label{sec:sufficiency}

We start proving the ``if'' part
in Theorem~\ref{thm:main} by considering two special cases
in Sections~\ref{sec:anticipationzero} and~\ref{sec:anticipationone}.
We then turn to the general case in Section~\ref{sec:stethering}.

\subsection{Deterministic encoders}
\label{sec:anticipationzero}

If $\XX(A_{G_0},A_{G_1},n_0,n_1)$ contains
a~$0\textrm{--}1$ vector, then a subgraph of~$G$ is
an~$(S(G),n_0,n_1)$-encoder with anticipation~$0$.
By Corollary~\ref{cor:lowerboundanticipation},
the existence of such a vector is also a necessary condition
for having a deterministic $(S(G),n_0,n_1)$-encoder.
If, in addition, $G$ has finite memory~$\mu$, then
the resulting encoder is $(\mu,0)$-definite
and, therefore,
$(\mu,0)$-sliding-block decodable for any tagging.

\subsection{Encoders with anticipation~$1$ obtained by state splitting}
\label{sec:anticipationone}

Suppose now that
$\XX(A_{G_0},A_{G_1},n_0,n_1)$ 
does not contain a~$0\textrm{--}1$ vector,
yet contains a vector~$\bldx = (x_u)_u$ such that
for each $\parity \in \Int{2}$, an application of one
$\bldx$-consistent state splitting round\footnote{%
Refer to~\cite[Ch.~5]{MRS} for the description of
the state-splitting algorithm and for the related terms used here.}
to $G_\parity$ (after deleting all states $u$ of~$G$ with $x_u = 0$)
results in an all-$1$ induced approximate eigenvector.
For $\parity \in \Int{2}$, let $\hat{\encoder}_\parity$ be
the resulting $(S(G_\parity),n_0,n_1)$-encoder.
Note that each state $u \in V(G)$ is transformed into
$x_u$ descendant states in each encoder~$\hat{\encoder}_\parity$;
denote those states by $(u,i)_\parity$, where $i \in \Int{x_u}$
(the order implied by the index~$i$ on the descendant states
of a given~$u$ can be arbitrary).

Next, construct the following graph~$\encoder$:
\[
V(\encoder) =
\Bigl\{
(u,i) \;:\; u \in V(G)
\quad \textrm{and} \quad  i \in \Int{x_u}
\Bigr\}
\; ,
\]
and endow $\encoder$ with
an edge $(u,i) \stackrel{a}{\rightarrow} (v,j)$
if and only if for some $\parity \in \Int{2}$,
the encoder $\hat{\encoder}_\parity$ contains
an edge $(u,i)_\parity \stackrel{a}{\rightarrow} (v,j)_\parity$.
It follows from the construction that
$\encoder_\parity = \hat{\encoder}_\parity$ for $\parity \in \Int{2}$.
In particular, from each $\encoder$-state there are~$n_\parity$ outgoing
edges with labels from~$\Sigma_\parity$, for each $\parity \in \Int{2}$.
Moreover, it can be readily seen that
$\FF_\encoder((u,i)) \subseteq \FF_G(u)$
for every $u \in V(G)$ and $i \in \Int{x_u}$.
Finally, $\encoder$ has anticipation~$1$:
if a word $w_1 w_2$ is generated in~$\encoder$ by a path $\pi$ from
state $(u,i) \in V(\encoder)$,
then the parent $G$-state~$u$ of~$(u,i)$
and the symbol~$w_1$ uniquely identify the parent $G$-state, $v$, of
the terminal state of the first edge in~$\pi$,
and the symbol~$w_2$ then uniquely identifies
the particular descendant state~$(v,j)$ of~$v$
in which that edge terminates.

In summary, $\encoder$ is
an $(S(G),n_0,n_1)$-encoder with anticipation~$1$.
Furthermore, if $G$ has finite memory~$\mu$, then~$\encoder$
is~$(\mu,1)$-definite and, therefore,
$(\mu,1)$-sliding-block decodable for any tagging.

\begin{example}
\label{ex:210rll}
We consider the $16$th power of the $(2,10)$-RLL constraint,
as found in the DVD standard~\cite[\S 1.7.3 and Example~5.7]{MRS}.
Let~$G$ be the graph presentation of that power,
where the states are numbered from~$0$ to~$10$,
and edges labeled with $16$-bit words that end with
a run of $0$s of length~$i \in \Int{11}$ terminate in state~$i$.
Also, let~$\Sigma_0$ (\resp, $\Sigma_1$) be
the set of all $16$-bit words of even (\resp, odd) parity
that satisfy the $(2,10)$-RLL constraint.
Then,
\[
A_{G_0} = 
\textrm{\small
$\left(
\renewcommand{\arraystretch}{0.8}
\arraycolsep0.7ex
\begin{array}{ccccccccccc}
42 & 28 & 19 & 12 &  8 &  6 &  5 &  4 &  3 &  2 &  1 \\
62 & 42 & 28 & 19 & 12 &  8 &  6 &  5 &  4 &  3 &  2 \\
90 & 62 & 42 & 28 & 19 & 12 &  8 &  6 &  5 &  4 &  3 \\
89 & 61 & 41 & 27 & 18 & 12 &  8 &  6 &  5 &  4 &  3 \\
88 & 60 & 40 & 26 & 17 & 11 &  8 &  6 &  5 &  4 &  3 \\
86 & 59 & 39 & 25 & 16 & 10 &  7 &  6 &  5 &  4 &  3 \\
82 & 57 & 38 & 24 & 15 &  9 &  6 &  5 &  5 &  4 &  3 \\
75 & 53 & 36 & 23 & 14 &  8 &  5 &  4 &  4 &  4 &  3 \\
65 & 46 & 32 & 21 & 13 &  7 &  4 &  3 &  3 &  3 &  3 \\
50 & 36 & 25 & 17 & 11 &  6 &  3 &  2 &  2 &  2 &  2 \\
29 & 21 & 15 & 10 &  7 &  4 &  2 &  1 &  1 &  1 &  1
\end{array}
\right)$}
\; \phantom{.}
\]
and
\[
A_{G_1} = 
\textrm{\small
$\left(
\arraycolsep0.7ex
\renewcommand{\arraystretch}{0.8}
\begin{array}{ccccccccccc}
41 & 29 & 21 & 15 & 10 &  7 &  4 &  2 &  1 &  1 &  1 \\ 
60 & 41 & 29 & 21 & 15 & 10 &  7 &  4 &  2 &  1 &  1 \\ 
87 & 60 & 41 & 29 & 21 & 15 & 10 &  7 &  4 &  2 &  1 \\ 
85 & 59 & 41 & 29 & 21 & 15 & 10 &  6 &  4 &  2 &  1 \\ 
82 & 57 & 40 & 29 & 21 & 15 & 10 &  6 &  3 &  2 &  1 \\ 
78 & 54 & 38 & 28 & 21 & 15 & 10 &  6 &  3 &  1 &  1 \\ 
73 & 50 & 35 & 26 & 20 & 15 & 10 &  6 &  3 &  1 &  0 \\ 
67 & 45 & 31 & 23 & 18 & 14 & 10 &  6 &  3 &  1 &  0 \\ 
59 & 39 & 26 & 19 & 15 & 12 &  9 &  6 &  3 &  1 &  0 \\ 
47 & 31 & 20 & 14 & 11 &  9 &  7 &  5 &  3 &  1 &  0 \\ 
28 & 19 & 12 &  8 & 6  &  5 &  4 &  3 &  2 &  1 &  0
\end{array}
\right)$}
\; .
\]
Running the algorithm in Figure~\ref{fig:fz} with
$A_0 = A_{G_0}$, $A_1 = A_{G_1}$, and
$\bldxi = 2 \cdot \bldone$
yields the result
\[
\bldx = 
( 1 \; 1 \; 2 \; 2 \; 2 \; 2 \; 1 \; 1 \; 1 \; 1 \; 0 )^\top \; ,
\]
for any~$n_0 \le 173$ and $n_1 \le 178$
(running the algorithm with larger values of~$n_0$ or~$n_1$
yields the all-zero vector).
This is also the vector obtained
when running the (ordinary) Franaszek algorithm
with~$A_G = A_{G_0} + A_{G_1}$,
$n = 351$, and~$\bldxi = 2 \cdot \bldone$.
In both~$G_0$ and~$G_1$ we can merge
states~$2\textrm{--}5$ into~$5$, states~$6\textrm{--}9$
into~$9$, and delete state~$10$
(see~\cite[\S 5.5.1]{MRS}),
resulting in graphs~$G_0'$ and~$G_1'$ with
\[
A_{G_0'} = 
\left(
\begin{array}{cccc}
42 & 28 &                45 &                14 \\
62 & 42 &                67 &                18 \\
86 & 59 &                90 &                22 \\
50 & 36 &                59 &                 9 
\end{array}
\right)
\; \phantom{,}
\]
and
\[
A_{G_1'} = 
\left(
\begin{array}{cccc}
41 & 29 &                53 &                 8 \\ 
60 & 41 &                75 &                14 \\ 
78 & 54 &               102 &                20 \\ 
47 & 31 &                54 &                16
\end{array}
\right)
\; ,
\]
and the respective
$(A_{G_0'},A_{G_1'},n_0{=}173,n_1{=}178)$-approximate eigenvector is
\[
\bldx' = ( 1 \; 1 \; 2 \; 1 )^\top \; .
\]
Both~$G_0'$ and~$G_1'$ can be split in one round
consistently with~$\bldx'$, resulting in the all-$1$
induced approximate eigenvector and,
therefore, in an~$(S(G),173,178)$-encoder.
The out-degree of the encoder that is actually used in the DVD
is~$2^p = 256$, where the set of input tags consists of
all $8$-bit tuples. Some of the input tags can be mapped
to two possible $16$-bit words
with different parities\footnote{%
There can be at most~$173 + 178 - 256 = 95$ input bytes of this type,
but in practice their number is slightly smaller.},
while the rest are mapped to unique $16$-bit words.\qed
\end{example}

\subsection{Construction using the stethering method}
\label{sec:stethering}

The technique used in Section~\ref{sec:anticipationone}
does not seem to generalize easily if the conditions therein---namely,
being able to split~$G_0$ and~$G_1$ in one round and ending up 
with an all-$1$ induced approximate eigenvector---do not hold.
In fact, due to the fact that the matrices~$G_0$ and~$G_1$
may be reducible, we may get stuck while attempting to split them.

\begin{example}
\label{ex:doesnotsplit}
Let $G$ be the graph with
$V(G) = \{ \alpha, \beta, \gamma \}$
whose even and odd subgraphs,
$G_0$ and $G_1$, are shown in
Figures~\ref{fig:G0doesnotsplit} and~\ref{fig:G1doesnotsplit}
(note that~$G_1$ is reducible).
All the edges in~$G$ are assumed to have distinct labels.
Assuming the ordering $\alpha < \beta < \gamma$
on the states, the adjacency matrices of the subgraphs are given by
\[
A_{G_0} =
\left(
\begin{array}{ccc}
0 & 1 & 0 \\
1 & 0 & 1 \\
1 & 1 & 1
\end{array}
\right)
\quad \textrm{and} \quad
A_{G_1} =
\left(
\begin{array}{ccc}
0 & 1 & 0 \\
1 & 0 & 1 \\
0 & 0 & 2
\end{array}
\right)
\; .
\]
It is easy to see that
$\bldx = (1 \; 2 \; 3)^\top$ is an eigenvector of both~$A_{G_0}$
and~$A_{G_1}$ associated with the Perron eigenvalue $n = 2$.
Yet the subgraph~$G_1$ cannot be split consistently with~$\bldx$.\qed
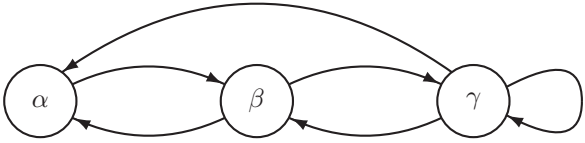
\begin{figure}[hbt]
\begin{center}
\thicklines
\setlength{\unitlength}{0.48mm}
\begin{picture}(170,40)(-10,-10)
    \multiput(000,000)(060,000){3}{\circle{20}}
    \qbezier(051,4.7)(030,14.5)(009,4.7)
    \put(050.6,4.7){\vector(2,-1){0}}
    \qbezier(051,-4.7)(030,-14.5)(009,-4.7)
    \put(009.4,-4.7){\vector(-2,1){0}}
    \qbezier(111,4.7)(090,14.5)(069,4.7)
    \put(110.6,4.7){\vector(2,-1){0}}
    \qbezier(111,-4.7)(090,-14.5)(069,-4.7)
    \put(069.4,-4.7){\vector(-2,1){0}}
    \qbezier(129.26,004)(150,015)(150,000)
    \qbezier(129.26,-04)(150,-15)(150,000)
    \put(129.26,-04){\vector(-2,1){0}}
    \qbezier(114,008)(060,046)(006,008)
    \put(006,008){\vector(-3,-2){0}}

    \put(000,000){\makebox(0,0){$\alpha$}}
    \put(060,000){\makebox(0,0){$\beta$}}
    \put(120,000){\makebox(0,0){$\gamma$}}
\end{picture}
\thinlines
\setlength{\unitlength}{1pt}
\end{center}
\caption{Subgraph $G_0$ for Example~\protect\ref{ex:doesnotsplit}.}
\label{fig:G0doesnotsplit}
\end{figure}
\begin{figure}[hbt]
\begin{center}
\thicklines
\setlength{\unitlength}{0.48mm}
\begin{picture}(170,30)(-10,-15)
    \multiput(000,000)(060,000){3}{\circle{20}}
    \qbezier(051,4.7)(030,14.5)(009,4.7)
    \put(050.6,4.7){\vector(2,-1){0}}
    \qbezier(051,-4.7)(030,-14.5)(009,-4.7)
    \put(009.4,-4.7){\vector(-2,1){0}}
    \put(070,000){\vector(1,0){40}}
    \qbezier(129.26,004)(150,015)(150,000)
    \qbezier(129.26,-04)(150,-15)(150,000)
    \put(129.26,-04){\vector(-2,1){0}}
    \qbezier(128,006)(160,025)(160,000)
    \qbezier(128,-06)(160,-25)(160,000)
    \put(128,-06){\vector(-2,1){0}}
    \put(000,000){\makebox(0,0){$\alpha$}}
    \put(060,000){\makebox(0,0){$\beta$}}
    \put(120,000){\makebox(0,0){$\gamma$}}
\end{picture}
\thinlines
\setlength{\unitlength}{1pt}
\end{center}
\caption{Subgraph $G_1$ for Example~\protect\ref{ex:doesnotsplit}.}
\label{fig:G1doesnotsplit}
\end{figure}
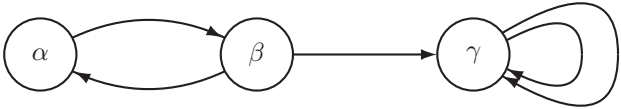
\end{example}

Moreover, in Appendix~\ref{sec:limitations} we present an example
where multiple rounds of state splitting are required,
which do end up with an all-$1$ approximate eigenvector, yet
there is no way one can match the descendant states in~$\encoder_0$ of
a given $G$-state with the respective descendant states in~$\encoder_1$ 
while maintaining finite anticipation.

Recognizing that the finite anticipation property is not
guaranteed even when the state-splitting algorithm is used
(at least in the manner we employed this algorithm
in Section~\ref{sec:anticipationone}), we resort to
a more general framework of designing encoders, which
includes the state-splitting algorithm and
the stethering design method of~\cite{AMR} as special cases
(see also~\cite{AGW} and~\cite[\S 6.2]{MRS}).
As we see, it will be rather easy to adapt the stethering method to
design bi-modal encoders, even though finite anticipation
can be guaranteed only under certain conditions.

Next we recall the stethering method, while tailoring it
to our setting.
Let~$G$ be a deterministic graph and $\{ \Sigma_0, \Sigma_1 \}$
be a partition of its label alphabet~$\Sigma$,
and let $\bldx = ( x_u )_{u \in V(G)}$ be
in $\XX(A_{G_0},A_{G_1},n_0,n_1)$. We assume that~$\bldx > \bldzero$,
or else remove the zero-weight states from~$G$
(namely, the states~$u$ for which $x_u = 0$).
For $u \in V(G)$, denote by
$\Sigma_\parity(u)$ the set of symbols from~$\Sigma_\parity$ that label
edges outgoing from~$u$.
For $u \in V(G)$ and
$a \in \Sigma_\parity(u)$, denote by~$\terminal(u;a)$
the terminal $G$-state of the unique edge outgoing
from~$u$ with label~$a$.

For $\parity \in \Int{2}$ and~$u \in V(G)$, let
\[
\Delta_\parity(u) =
\Bigl\{
(a,j) \;:\; a \in \Sigma_\parity(u)
\quad \textrm{and} \quad
j \in \Int{x_{\terminal(u;a)}}
\Bigr\}
\; .
\]
Since $\bldx \in \XX(A_{G_\parity},n_\parity)$ we have
$|\Delta_\parity(u)| = (A_{G_\parity} \bldx)_u \ge n_\parity x_u$.
Thus, we can partition
(a subset of) $\Delta_\parity(u)$ into~$x_u$ subsets
\begin{equation}
\label{eq:partition}
\Delta_\parity^{(0)}(u), \; \Delta_\parity^{(1)}(u), \; \ldots, \;
\Delta_\parity^{(x_u-1)}(u) \; ,
\end{equation}
such that $|\Delta_\parity^{(i)}(u)| = n_\parity$ for each~$i$.
In what follows, we fix such a partition.

Next, construct the following graph~$\encoder$:
\[
V(\encoder) =
\Bigl\{
(u,i) \;:\; u \in V(G)
\quad \textrm{and} \quad  i \in \Int{x_u}
\Bigr\}
\; ,
\]
and for each $\parity \in \Int{2}$, $u \in V(G)$,
$i \in \Int{x_u}$, and $(a,j) \in \Delta_\parity^{(i)}(u)$,
we endow $\encoder$ with an edge
$(u,i) \stackrel{a}{\rightarrow} (\terminal(u;a),j)$.

\begin{proposition}
\label{prop:AGW}
The constructed graph~$\encoder$ is an~$(S(G),n_0,n_1)$-encoder.
\end{proposition}

\begin{proof}
First, by construction, the number of outgoing edges from~$(u,i)$
with labels from~$\Sigma_\parity$ is
$|\Delta_\parity^{(i)}(u)| = n_\parity$, for each $\parity \in \Int{2}$.

Secondly, let
\begin{equation}
\label{eq:path}
(u_0,i_0)
\stackrel{w_1}{\longrightarrow} (u_1,i_1)
\stackrel{w_2}{\longrightarrow} (u_2,i_2)
\stackrel{w_3}{\longrightarrow} \ldots
\stackrel{w_\ell}{\longrightarrow} (u_\ell,i_\ell)
\end{equation}
be a path in~$\encoder$.
By construction, $u_{m+1} = \terminal(u_m;w_{m+1})$
for every $m \in \Int{\ell}$. Hence,
$S(\encoder) \subseteq S(G)$.

It remains to show that~$\encoder$ is lossless.
Consider the word $\bldw = w_1 w_2 \ldots w_\ell$
generated by the path~(\ref{eq:path}). We show that
the knowledge of~$\bldw$, $(u_0,i_0)$, and~$(u_\ell,w_\ell)$
uniquely determines the rest of the states along the path.
Since~$G$ is deterministic, the component~$u_m$ of each
state~$(u_m,i_m)$ along the path is uniquely determined.
Suppose by induction that~$(u_{m'},i_{m'})$ has been uniquely
determined for every $m' > m$, and let~$\parity_{m+1}$ be
the parity of~$w_{m+1}$ (i.e., $w_{m+1} \in \Sigma_{\parity_{m+1}}$).
Since the subsets in~(\ref{eq:partition}) are disjoint
for~$(u,\parity) = (u_m,\parity_{m+1})$, there is
a unique index $i \in \Int{x_{u_m}}$ 
for which $(w_{m+1},i_{m+1}) \in \Delta_{\parity_{m+1}}^{(i)}(u_m)$;
that index must be~$i = i_m$.
\end{proof}

The number of states of the constructed encoder~$\encoder$
(before any possible merging of states)
equals the sum, $\| \bldx \|_1$, of the entries of~$\bldx$.
Thus, with this construction, we can obtain an encoder~$\encoder$ 
such that
\[
|V(\encoder)| \le
\min_{\bldx \in \XX(A_{G_0},A_{G_1},n_0,n_1)}
\| \bldx \|_1
\]
(compare with the lower bound of Corollary~\ref{cor:lowerboundstates}).

In stethering encoders, the subsets in~(\ref{eq:partition})
have a particular structure which we describe next.
For $\parity \in \Int{2}$ and $u \in V(G)$,
assume some ordering on the elements of~$\Sigma_\parity(u)$.
For $a \in \Sigma_\parity(u)$, define $\phi_\parity(u;a)$ by
\[
\phi_\parity(u;a) =
\sum_{b \in \Sigma_\parity(u) \;:\; b < a} x_{u_b} \; ,
\]
where the sum is zero on an empty set.
For
$i \in \Int{x_u}$, let
\begin{eqnarray}
\lefteqn{
\Delta_\parity^{(i)}(u) =
\Bigl\{
(a,j)
\;:\;  a \in \Sigma_\parity(u) \; ,
\;\;
j \in \Int{x_{\terminal(u;a)}} \; ,
} \makebox[8ex]{} \nonumber \\
&&
\label{eq:partitionstethering}
\textrm{and} \;\;
i \, n_\parity \le \phi_\parity(u;a) + j < (i{+}1) n_\parity
\Bigr\}
\; .
\end{eqnarray}
In other words,
for each $u, v \in V(G)$, $i \in \Int{x_u}$, $j \in \Int{x_v}$,
and $a \in \Sigma_\parity(u)$, we endow~$\encoder$ with an edge
$(u,i) \stackrel{a}{\rightarrow} (v,j)$, if and only if
$v = \terminal(u;a)$ and
\[
i \, n_\parity \le \phi_\parity(u;a) + j < (i{+}1) n_\parity \; .
\]

This construction is illustrated in Figure~\ref{fig:stethering},
for a given $G$-state~$u$ and parity~$\parity \in \Int{2}$.
The boxes in the top row in the figure represent
the ``descendants'' of state~$u$, namely,
the states~$(u,i)$, for~$i \in \Int{x_u}$,
and the width of each box in the top row is one unit.
The outgoing edges from each state~$(u,i)$ are shown as downward
arrows, along with their labels, and an assignment of input tags,
$(\parity,0), (\parity,1), \ldots, (\parity,n_\parity{-}1)$,
is shown above the top row. The boxes at the bottom row 
are $1/n_\parity$ units wide
and represent the terminal states of the edges.
The respective elements of~$\Delta_\parity(u)$ are written
just below the bottom row, where we have also shown
their grouping into the subsets~(\ref{eq:partition})
defined by~(\ref{eq:partitionstethering}).
The double vertical lines group the edges according to their labels.
So, for example, according to the figure,
there is an outgoing edge labeled~$a'$
and tagged by~$(\parity,0)$ from~$(u,x_u{-}1)$
to~$(v',x_{v'}{-}1)$, and that edge corresponds to the element
$(a',x_{v'}{-}1) \in \Delta_\parity(u)$.
\begin{figure*}[hbt]
\[
\begin{array}{cccccccccc}
\multicolumn{1}{c}{\scriptstyle (\parity,0)} &
\multicolumn{1}{c@{\scriptstyle \ldots}}{\scriptstyle (\parity,1)} &
\multicolumn{1}{c}{\scriptstyle (\parity,n_\parity{-}1)} &
\multicolumn{1}{c}{\scriptstyle (\parity,0)} &
\multicolumn{1}{c@{\scriptstyle \ldots}}{\scriptstyle (\parity,1)} &
\multicolumn{1}{c}{\scriptstyle (\parity,n_\parity{-}1)} &
\multicolumn{1}{c}{\scriptstyle (\parity,0)} &
\multicolumn{1}{c@{\scriptstyle \ldots}}{\scriptstyle (\parity,1)} &
\multicolumn{1}{c}{\scriptstyle (\parity,n_\parity{-}1)} &
\multicolumn{1}{c}{} \\
\cline{1-9}
\multicolumn{3}{||c|}{(u,0)} &
\multicolumn{3}{c|}{\cdots} &
\multicolumn{3}{c||}{(u,x_u{-}1)}  \\
\cline{1-9} 
\phantom{a} \downarrow a &
\phantom{a} \downarrow a &
\phantom{a} \downarrow a &
\phantom{a} \downarrow a &
\phantom{a'} \downarrow a' &
\phantom{a'} \downarrow a' &
\phantom{a'} \downarrow a' &
\phantom{a''} \downarrow a'' &
\phantom{a''} \downarrow a'' &
\\
\cline{1-10}
\multicolumn{1}{||c|}{\makebox[2.9em][c]{$\scriptstyle (v{,}0)$}} &
\multicolumn{1}{c|}{\makebox[2.9em][c]{$\scriptstyle (v{,}1)$}} &
\multicolumn{1}{c|}{\makebox[2.9em][c]{$\cdots$}} &
\multicolumn{1}{c||}{\makebox[2.9em][c]{$\scriptstyle (v{,}x_v{-}1)$}} &
\multicolumn{1}{c|}{\makebox[2.9em][c]{$\scriptstyle (v'{,}0)$}} &
\multicolumn{1}{c|}{\makebox[2.9em][c]{$\cdots$}} &
\multicolumn{1}{c||}{\makebox[2.9em][c]{$\scriptstyle
                                            (v'{,}x_{v'}{-}1)$}} &
\multicolumn{1}{c|}{\makebox[2.9em][c]{$\scriptstyle (v''{,}0)$}} &
\multicolumn{1}{c|}{\makebox[2.9em][c]{$\scriptstyle (v''{,}1)$}} &
\multicolumn{1}{c||}{\makebox[2.9em][c]{$\cdots$}} \\
\cline{1-10}
\multicolumn{1}{c}{\scriptstyle (a{,}0)} &
\multicolumn{1}{c}{\scriptstyle (a{,}1)} &
\multicolumn{1}{c}{\cdots} &
\multicolumn{1}{c}{\scriptstyle (a{,}x_v{-}1)} &
\multicolumn{1}{c}{\scriptstyle (a'{,}0)} &
\multicolumn{1}{c}{\cdots} &
\multicolumn{1}{c}{\scriptstyle (a'{,}x_{v'}{-}1)} &
\multicolumn{1}{c}{\scriptstyle (a''{,}0)} &
\multicolumn{1}{c}{\scriptstyle (a''{,}1)} &
\multicolumn{1}{c}{\cdots} \\
\multicolumn{3}{c}{\longleftarrow \hfill
\Delta_\parity^{(0)}(u) \hfill \longrightarrow} &
\multicolumn{3}{c}{\longleftarrow \hfill
\cdots \hfill \longrightarrow} &
\multicolumn{3}{c}{\longleftarrow \hfill
\Delta_\parity^{(x_u-1)}(u) \hfill \longrightarrow} &
\end{array}
\]
\caption{Descendants of a $G$-state~$u$
in a subgraph~$\encoder_\parity$ of a stethering encoder.}
\label{fig:stethering}
\end{figure*}

Stethering encoders can have finite anticipation
(and be sliding-block decodable if~$G$ has finite memory),
provided that there is sufficient margin between
the target encoder rate $p:q$ and
the maximal coding ratio~$\ratio(G,q)$
(as defined in~(\ref{eq:codingratio})).
We demonstrate this next.

Suppose that $\bldx \in \XX(A_{G_0},A_{G_1},n_0{+}1,n_1{+}1)$
(namely, we assume even and odd out-degrees larger by~$1$
than targeted).
Using~$\bldx$, we first construct
a stethering $(S(G),n_0{+}1,n_1{+}1)$-encoder~$\encoder^*$
and assign the input tags
$(\parity,0), (\parity,1), \ldots , (\parity,n_\parity)$
to the outgoing edges from each state, as
in Figure~\ref{fig:stethering}. Then, from $\encoder^*$
we form a punctured $(S(G),n_0,n_1)$-encoder $\encoder$
by deleting all edges in~$\encoder$ tagged
by either~$(0,n_0)$ or~$(1,n_1)$.

We have the following result (compare the guaranteed
upper bound on~$\anticipation(\encoder)$ to the lower bound
in Corollary~\ref{cor:lowerboundanticipation}).

\begin{theorem}
\label{thm:stethering}
Let~$G$ be
a deterministic graph and let~$n_0$ and~$n_1$ be positive integers
such that
$\XX(A_{G_0},A_{G_1},n_0{+}1,n_1{+}1) \ne \emptyset$.
Then, there is an~$(S(G),n_0,n_1)$-encoder~$\encoder$, obtained by
the (punctured) stethering method, such that
$\anticipation(\encoder) \le a$, where
\[
a = 1 +
\min_{\bldx \in \XX(A_{G_0},A_{G_1},n_0{+}1,n_1{+}1)}
\Bigl\{ \,
\lceil \log_{n+1} \| \bldx \|_\infty \rceil \, \Bigr\}
\]
and~$n = \min \{ n_0, n_1 \}$.
Furthermore, if $G$ has finite memory~$\mu$, then~$\encoder$ is
$(\mu,a)$-definite,
and hence any tagged $(S(G),n_0,n_1)$-encoder based on~$\encoder$ is
$(\mu,a)$-sliding-block decodable.
\end{theorem}

\begin{proof}
The proof is essentially the same as that
of Proposition~3 in~\cite{AMR}, and we repeat it here
(with the required modifications to handle
the bi-modal setting) for completeness.

Let $\Stether_\parity(u)$ denote
Figure~\ref{fig:stethering} drawn for a given
$\parity \in \Int{2}$ and~$u \in V(G)$.
Consider a path
\[
\pi = 
(u_0,i_0)
\stackrel{w_1}{\longrightarrow} (u_1,i_1)
\stackrel{w_2}{\longrightarrow} (u_2,i_2)
\stackrel{w_3}{\longrightarrow} \ldots
\]
in the encoder~$\encoder^*$ (obtained prior to the puncturing),
and let $(\parity_1,s_1),(\parity_2,s_2),(\parity_3,s_3),\ldots$
be the respective sequence of input tags, where
$w_m \in \Sigma_{\parity_m}$. Envision an array~$\Array$
which is constructed as follows.
Start with the figure~$\Stether_{\parity_1}(u_0)$;
the edges labeled~$w_1$ in the figure
terminate in the descendant states of~$u_1$,
which appear as~$x_{u_1}$ boxes in
the bottom row of~$\Stether_{\parity_1}(u_0)$.
Up to down-scaling by a factor of~$n_{\parity_1}$,
these boxes are identical to the top row in~$\Stether_{\parity_2}(u_1)$.
So, in~$\Array$, superimpose a down-scaled copy of
$\Stether_{\parity_2}(u_1)$ so that its top row coincides with
the descendant states of~$u_1$ in~$\Stether_{\parity_1}(u_0)$.
Proceed in this manner by placing in~$\Array$ a copy
of $\Stether_{\parity_3}(u_2)$, down-scaled by a factor
of~$n_{\parity_1} n_{\parity_2}$, so that its top row coincides
with the descendant states of~$u_2$ in the bottom row of
the (already inserted) down-scaled copy of~$\Stether_{\parity_2}(u_1)$.
And so on.

The path~$\pi$ can be seen as a vertical line in~$\Array$
whose abscissa (i.e., the distance from the left margin of~$\Array$)
has the mixed-base representation
$i_0.s_{\parity_1} s_{\parity_2} s_{\parity_3} \ldots$,
where $i_0 \in \Int{x_u}$ and
$s_{\parity_m} \in \Int{n_{\parity_m}{+}1}$.
In other words, that abscissa equals
\[
i_0 +
\frac{s_{\parity_1}}{n_{\parity_1}{+}1}
+ \frac{s_{\parity_2}}{(n_{\parity_1}{+}1)(n_{\parity_2}{+}1)}
+ \ldots
\]
Due to the down-scaling process used to construct~$\Array$,
a decoder can narrow down the uncertainty of that abscissa
for each received symbol.
Specifically, merely by the knowledge of~$u_0$ that abscissa
must be in the real interval $[0,x_{u_0})$ (which is the full
width of~$\Array$). Upon receiving~$w_1$, 
the length of that uncertainty (open) interval
shrinks to $x_{u_1}/(n_{\parity_1}{+}1)$; then~$w_2$
reduces it to $x_{u_2}/((n_{\parity_1}{+}1)(n_{\parity_2}{+}1))$,
and so forth. Hence, when the length~$\ell$ of the path is such that
\begin{equation}
\label{eq:abscissa}
\frac{x_{u_\ell}}{\prod_{m=1}^\ell (n_{\parity_m}+1)}
\le
\frac{1}{(n_{\parity_1}{+}1)(n_{\parity_2}{+}1)} \; ,
\end{equation}
the length of the uncertainty interval reduces to
at most the right-hand side of~(\ref{eq:abscissa}).
At this point, the numerator in the expression
\[
\frac{(n_{\parity_2}{+}1) s_{\parity_1} + s_{\parity_2}}%
                                {(n_{\parity_1}{+}1)(n_{\parity_2}{+}1)}
\]
(for the abscissa point $0.s_{\parity_1} s_{\parity_2}$)
can be determined up to~$\pm 1$. Yet since the puncturing disallows
$s_{\parity_2}$ to take the value~$n_{\parity_2}$, this means
that $s_{\parity_1}$ is uniquely determined.
It is easy to see that~(\ref{eq:abscissa}) is satisfied
for $\ell = a + 1 = 2 + \lceil \log_{n+1} \| \bldx \|_\infty \rceil$,
where $n = \min \{ n_0, n_1 \}$, thereby proving the claimed
upper bound on~$\anticipation(\encoder)$.
Moreover,
if~$G$ has finite memory~$\mu$, then the decoder can recover~$u_0$
by looking at a window of~$\mu$ past symbols, i.e.,
$\encoder$ is~$(\mu,a)$-definite.
\end{proof}

Recall that $n_{\max}(G,q)$ is the largest integer~$n$ for which
$(S(G^q),n,n)$-encoders exist.
If we use the punctured stethering method to construct
rate $p:q$ bi-modal encoders, then we need
to have $2^p + 1 \le n_{\max}(G,q)$. This inequality is satisfied
whenever
\[
\frac{p}{q} \le
\frac{\log_2 n_{\max}(G,q)}{q} - \frac{\log_2 (1+2^{-p})}{q} \; ,
\]
which, in turn, is satisfied whenever
\[
\frac{p}{q} \le 
\ratio(G,q) - \frac{\log_2 e}{2^p q}
\]
(see~(\ref{eq:codingratio})).
We conclude that finite anticipation (and sliding-block decodability
when~$G$ has finite memory) can be guaranteed with a rate penalty
of (no more than) $(\log_2 e)/(2^p q)$.

It is still an open problem whether
finite anticipation can be guaranteed for any~$(n_0,n_1)$ for which
$\XX(A_{G_0},A_{G_1},n_0,n_1) \ne \emptyset$.

\subsection{Extension to non-partition covers}
\label{sec:non-partition}

We refer now to a more general setting
where the even subset $\Sigma_0$ and the odd subset $\Sigma_1$
of the constraint alphabet~$\Sigma$ are not necessarily disjoint,
but their union is still $\Sigma$
(recall Footnote~\ref{footnote:disparity} for an application of
this scenario). Thus, given a graph~$G$,
the subgraphs $G_0$ and $G_1$ may share edges.
It turns out that most of our results hold also for this setting.
In particular, the necessary condition stated in
Theorem~\ref{thm:necessity} holds as is,
since the proof of the theorem does not assume that
$\Sigma_0$ and $\Sigma_1$ are disjoint.

As for the sufficient condition shown in Section~\ref{sec:stethering},
the proof of Proposition~\ref{prop:AGW} holds as long as
the partitions~(\ref{eq:partition}) satisfy the following
condition for every
$u \in V(G)$, $a \in \Sigma_0(u) \cap \Sigma_1(u)$,
and $j \in \Int{x_{\terminal(u;a)}}$:
\begin{equation}
\label{eq:consistency}
(a,j) \in \Delta_0^{(i)}(u) \cap \Delta_1^{(i')}(u)
\quad \Longrightarrow \quad i = i'
\end{equation}
(this guarantees the uniqueness of the index $i = i_m$
in the last step of the proof of Proposition~\ref{prop:AGW},
even when the parity $\parity_{m+1}$ is not uniquely determined by
$w_{m+1}$).
Assuming without loss of generality that $n_0 \le n_1$,
we can apply the following strategy to guarantee
the condition~(\ref{eq:consistency}). We first select
an arbitrary partition~(\ref{eq:partition}) for $\parity = 0$
(where $|\Delta_0^{(i)}(u)| = n_0$ for each $i \in \Int{x_u}$).
Then, for each $u \in V(G)$, $a \in \Sigma_0(u) \cap \Sigma_1(u)$,
and $i \in \Int{x_u}$,
we let all elements $(a,j) \in \Delta_0^{(i)}(u)$
be also elements of $\Delta_1^{(i)}(u)$.
Finally, we fill in each subset $\Delta_1^{(i)}(u)$ from
the remaining elements of $\Delta_1(u)$
to have $|\Delta_1^{(i)}(u)| = n_1$.

We note that, in general,
the stethering partition~(\ref{eq:partitionstethering})
(which is used in the proof of Theorem~\ref{thm:stethering})
may be inconsistent with the condition~(\ref{eq:consistency}).
Still, we can always get consistency when $n_0 = n_1$,
by selecting the ordering on $\Sigma_0(u)$
and $\Sigma_1(u)$ so that the elements 
in the intersection $\Sigma_0(u) \cap \Sigma_1(u)$ precede the rest.

\section{Variable-length encoders}
\label{sec:vle}

So far in this work,
we considered bi-modal \emph{fixed-length} encoders
at a fixed rate $p:q$, where all tags have
the same length~$p$, and all labels have the same length~$q$
(under the formulation of $(S,n)$-encoders, these lengths are~$1$).
On the other hand, most ad-hoc constructions
for parity-preserving encoders that were proposed have
\emph{variable length} (yet still with a \emph{fixed coding ratio}).
In this section, we show through examples that
the flexibility of having variable-length (tags and) labels
strictly increases the attainable range of 
the coding ratios of parity-preserving encoders,
compared to the fixed-length case.
A more thorough study of parity-preserving variable-length encoders is
deferred to a subsequent work~\cite{RS}.
For more on variable-length encoders
(in the non-parity-preserving setting), see~%
\cite{AFKM},
\cite{Beal1},
\cite{Beal2},
\cite{Franaszek1},
\cite{Franaszek2},
\cite{HMS},
\cite[\S 6.4]{MRS}.

\begin{example}
\label{ex:vle}
Let~$S$ be the constraint over $\Sigma = \{ a, b, c, d \}$
that is presented by the graph~$G$ in Figure~\ref{fig:twostates}.
Consider the graph~$\encoder$ in Figure~\ref{fig:vle},
which has one state~$\alpha$ and three edges: an edge labeled~$a$
and two edges of length~$2$, with labels $bd$ and~$cd$
(namely, the length of an edge is the length of its label).
\begin{figure}[hbt]
\begin{center}
\thicklines
\setlength{\unitlength}{0.48mm}
\begin{picture}(080,035)(-35,-15)
    \put(000,000){\circle{20}}

    \qbezier(-9.26,004)(-30,015)(-30,000)
    \qbezier(-9.26,-04)(-30,-15)(-30,000)
    \put(-9.26,004){\vector(2,-1){0}}

    \qbezier(9.26,004)(030,015)(030,000)
    \qbezier(9.26,-04)(030,-15)(030,000)
    \put(9.26,004){\vector(-2,-1){0}}
    \qbezier(7.14,007)(040,025)(040,000)
    \qbezier(7.14,-07)(040,-25)(040,000)
    \put(7.14,007){\vector(-3,-2){0}}

\put(000,000){\makebox(0,0){$\alpha$}}

\put(-32,000){\makebox(0,0)[r]{$a$}}
\put(028,001){\makebox(0,0)[r]{$bd$}}
\put(042,001){\makebox(0,0)[l]{$cd$}}
\end{picture}
\thinlines
\setlength{\unitlength}{1pt}
\end{center}
\caption{Variable-length encoder~$\encoder$ for
the constraint presented by Figure~\protect\ref{fig:twostates}.}
\label{fig:vle}
\end{figure}
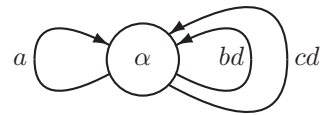
For the purpose of defining the words that can be generated
by a variable-length graph such as~$\encoder$, we view each
length-$\ell$ edge as if it were a path of length~$\ell$
(whose edges are connected by additional $\ell{-}1$ dummy states).
Doing so, it is easy to see that every word that can be generated
by~$\encoder$ can also be generated from state~$\alpha$
in the graph~$G$ of Figure~\ref{fig:twostates}.
The graph~$\encoder$ is \emph{deterministic} in the sense
that the set of labels is \emph{prefix-free}:
no label is a prefix of any other label. Hence, a word
generated by~$\encoder$ uniquely identifies the path that generates it.

We now assign tags over the (base tag) alphabet
$\Upsilon = \{ 0, 1 \}$ to the edges (labels) of~$\encoder$,
as shown in Table~\ref{tab:vle}.
\begin{table}[hbt]
\caption{Possible tag assignment for
the encoder in Figure~\ref{fig:vle}.}
\label{tab:vle}
\[
\begin{array}{rcl}
0  & \leftrightarrow & a  \\
10 & \leftrightarrow & bd \\
11 & \leftrightarrow & cd \\
\end{array}
\]
\end{table}
We get in this manner an encoder that has a \emph{coding ratio} of~$1$:
the coding rate is $1:1$ when the input tag is~$0$, 
and $2:2$ when the input tag starts with a~$1$.
Thus, this encoder is capacity-achieving.
Moreover, this tag assignment is parity-preserving
with respect to the partition $\{ \Sigma_0, \Sigma_1 \}$
defined in~(\ref{eq:partitioning}).
In contrast, we showed in Example~\ref{ex:twostates}
that, for this partition,
a coding ratio of~$1$ cannot be achieved 
by any bi-modal (and, \emph{a fortiori}, parity-preserving)
fixed-length encoder.\qed
\end{example}

\begin{example}
\label{ex:vlealt}
Considering the same constraint~$S$ as in the previous example,
the graph~$\encoder'$ in Figure~\ref{fig:vlealt} presents
another (untagged) variable-length encoder.
The coding rate at state~$\alpha'$
is $3:3$, as it has eight outgoing edges with labels in~$\Sigma^3$,
and the coding rate at~$\alpha''$ and at~$\beta$ is $2:2$,
as each state has four outgoing edges labeled from $\Sigma^2$;
the coding ratio at each state is therefore~$1$, 
making~$\encoder'$ capacity-achieving.
However, $\encoder'$ is not deterministic
(there are two edges labeled $bda$ and two labeled $cda$
outgoing from state~$\alpha'$,
two edges labeled $aa$ outgoing from~$\alpha''$,
and two labeled $da$ from state~$\beta$).
Nevertheless, $\encoder'$ has finite anticipation
and is therefore lossless:
the first symbol of a label uniquely determines
the length of the label as well as the initial state,
and a label and the first symbol of the next label 
within a sequence uniquely determine the edge.
One possible assignment of tags
(over the alphabet $\Upsilon = \{ 0, 1 \}$)
to the edges of~$\encoder'$ is shown in Table~\ref{tab:vlealt}.
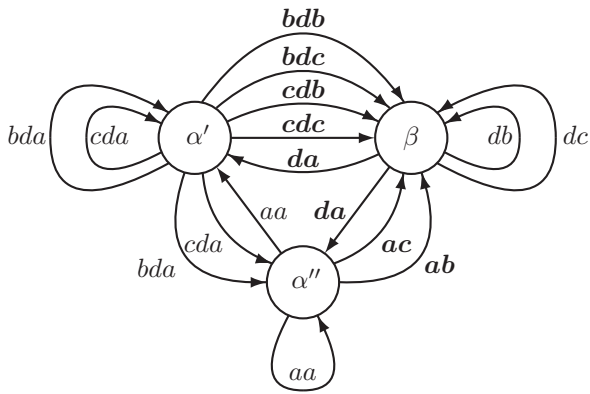
\begin{figure}[hbt]
\begin{center}
\thicklines
\setlength{\unitlength}{0.48mm}
\begin{picture}(150,110)(-45,-70)
    \multiput(000,000)(060,000){2}{\circle{20}}
    \put(030,-40){\circle{20}}
    \qbezier(051,4.5)(030,14.5)(009,4.5)
    \put(051,4.5){\vector(2,-1){0}}
    \qbezier(054,8.3)(030,28.8)(006,8.3)
    \put(054,8.3){\vector(5,-4){0}}
    \qbezier(058,9.9)(030,48.0)(002,9.9)
    \put(058,9.9){\vector(2,-3){0}}
    \qbezier(051,-4.5)(030,-14.5)(009,-4.5)
    \put(009,-4.5){\vector(-2,1){0}}

    \put(010,000){\vector(1,0){40}}

    \qbezier(-9.26,004)(-30,015)(-30,000)
    \qbezier(-9.26,-04)(-30,-15)(-30,000)
    \put(-9.26,004){\vector(2,-1){0}}
    \qbezier(-7.14,007)(-40,025)(-40,000)
    \qbezier(-7.14,-07)(-40,-25)(-40,000)
    \put(-7.14,007){\vector(3,-2){0}}

    \qbezier(69.26,004)(090,015)(090,000)
    \qbezier(69.26,-04)(090,-15)(090,000)
    \put(69.26,004){\vector(-2,-1){0}}
    \qbezier(67.14,007)(100,025)(100,000)
    \qbezier(67.14,-07)(100,-25)(100,000)
    \put(67.14,007){\vector(-3,-2){0}}

    \put(024,-32){\vector(-3,4){18}}
    \put(054,-08){\vector(-3,-4){18}}

    \qbezier(040,-40)(73.4,-41.3)(62.8,-9.6)
    \put(62.8,-9.6){\vector(-1,4){0}}
    \qbezier(38.6,-34.8)(56.2,-28.2)(57.8,-9.6)
    \put(57.8,-9.6){\vector(1,6){0}}

    \qbezier(020,-40)(-13.4,-41.3)(-2.8,-9.6)
    \put(020,-40){\vector(1,0){0}}
    \qbezier(21.4,-34.8)(3.8,-28.2)(2.2,-9.6)
    \put(21.4,-34.8){\vector(3,-1){0}}

    \qbezier(034,-49.26)(045,-70)(030,-70)
    \qbezier(026,-49.26)(015,-70)(030,-70)
    \put(034,-49.26){\vector(-1,2){0}}

\put(001,001){\makebox(0,0){$\alpha'$}}
\put(060,000){\makebox(0,0){$\beta$}}
\put(031,-39){\makebox(0,0){$\alpha''$}}

\put(-41,001){\makebox(0,0)[r]{$bda$}}
\put(-29,001){\makebox(0,0)[l]{$cda$}}
\put(030,030.6){\makebox(0,0)[b]{$\mathbf{bdb}$}}
\put(030,020.2){\makebox(0,0)[b]{$\mathbf{bdc}$}}
\put(030,011.0){\makebox(0,0)[b]{$\mathbf{cdb}$}}
\put(030,001.5){\makebox(0,0)[b]{$\mathbf{cdc}$}}
\put(030,-08.0){\makebox(0,0)[b]{$\mathbf{da}$}}
\put(088,001){\makebox(0,0)[r]{$db$}}
\put(102,001){\makebox(0,0)[l]{$dc$}}
\put(042,-20){\makebox(0,0)[r]{$\mathbf{da}$}}
\put(018,-21){\makebox(0,0)[l]{$aa$}}
\put(072,-35){\makebox(0,0)[r]{$\mathbf{ab}$}}
\put(060,-30){\makebox(0,0)[r]{$\mathbf{ac}$}}
\put(-16,-36){\makebox(0,0)[l]{$bda$}}
\put(-03,-29){\makebox(0,0)[l]{$cda$}}
\put(030,-67){\makebox(0,0)[b]{$aa$}}
\end{picture}
\thinlines
\setlength{\unitlength}{1pt}
\end{center}
\caption{Second variable-length encoder~$\encoder'$ for
the constraint presented by Figure~\protect\ref{fig:twostates}.}
\label{fig:vlealt}
\end{figure}
\begin{table}[hbt]
\caption{Possible tag assignment for
the encoder in Figure~\ref{fig:vlealt}.}
\label{tab:vlealt}
\[
\begin{array}{rcccrcccrcc}
\multicolumn{3}{c}{\mathrm{State} \; \alpha'}  &&
\multicolumn{3}{c}{\mathrm{State} \; \alpha''} &&
\multicolumn{3}{c}{\mathrm{State} \; \beta}    \\
\cline{1-3} \cline{5-7} \cline{9-11}
000, 011 & \leftrightarrow   & bda          &\;\;&
00, 11   & \leftrightarrow   & aa           &\;\;&
01, 10   & \leftrightarrow   & \mathbf{da}  \\
101, 110 & \leftrightarrow   & cda          &\;\;&
01       & \leftrightarrow   & \mathbf{ac}  &\;\;&
00       & \leftrightarrow   & db           \\
001      & \leftrightarrow   & \mathbf{bdb} &\;\;&
10       & \leftrightarrow   & \mathbf{ab}  &\;\;&
11       & \leftrightarrow   & dc           \\
010      & \leftrightarrow   & \mathbf{bdc} &&&&&& \\
100      & \leftrightarrow   & \mathbf{cdb} &&&&&& \\
111      & \leftrightarrow   & \mathbf{cdc} &&&&&& \\
\end{array}
\]
\end{table}

Consider now the partition~$\{ \Sigma_0, \Sigma_1 \}$
defined in~(\ref{eq:partitioningalt}).
The labels in boldface in Figure~\ref{fig:vlealt}
are the odd labels with respect to this partition.
It can be readily verified that the tag assignment
in Table~\ref{tab:vlealt} is parity-preserving.
In contrast, recall that we have shown in
Example~\ref{ex:twostatesalt}
that there is no bi-modal fixed-length encoder
at a coding ratio of~$1$ for this constraint and this partition.

The encoder in Figure~\ref{fig:vlealt}
can be obtained by first splitting state~$\alpha$ in
Figure~\ref{fig:twostates} into states~$\alpha'$ and $\alpha''$
which inherit, respectively,
the outgoing edge sets $\{ b, c \}$ and $\{ a \}$.
The resulting graph is an (ordinary) $(S,2)$-encoder,
yet, for the above partition of~$\Sigma$,
the parities of the two outgoing edges from each state are the same.
We then replace the outgoing edges from state~$\alpha'$
with the eight paths of length~$3$ that start at that state;
similarly, we replace the outgoing edges from
each of the states~$\alpha''$ and~$\beta$ with the four paths of
length~$2$ that start at the state.\qed
\end{example}

To summarize, for the two different partitions,
(\ref{eq:partitioning}) and~(\ref{eq:partitioningalt}),
of the alphabet $\Sigma = \{ a, b, c , d \}$,
Examples~\ref{ex:vle} and~\ref{ex:vlealt} present
respective (capacity-achieving)
parity-preserving variable-length encoders with a coding ratio of~$1$:
the first encoder is deterministic, while the other is not.
In fact, we show in~\cite{RS} that
for the partition~(\ref{eq:partitioningalt}),
one cannot achieve a coding ratio of~$1$ by
any deterministic parity-preserving variable-length encoder
(unless one uses a degenerate base tag alphabet
containing only even symbols).

On the other hand, there exists such an encoder
under some relaxation of the notion of fixed coding ratio,
following the encoding model considered in~\cite{HMS}:
the tagged encoder~$\encoder^\circ$ in Figure~\ref{fig:vleHMS}
maintains a coding ratio of~$1$ \emph{along each cycle}.
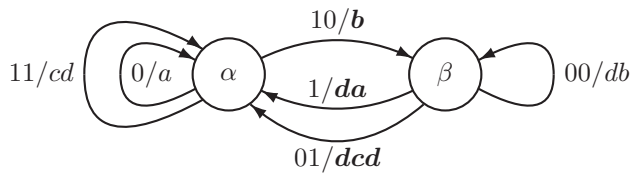
\begin{figure}[hbt]
\begin{center}
\thicklines
\setlength{\unitlength}{0.48mm}
\begin{picture}(150,050)(-45,-30)
    \multiput(000,000)(060,000){2}{\circle{20}}
    \qbezier(051,4.5)(030,14.5)(009,4.5)
    \put(051,4.5){\vector(2,-1){0}}
    \qbezier(051,-4.5)(030,-14.5)(009,-4.5)
    \put(009,-4.5){\vector(-2,1){0}}
    \qbezier(054,-8.3)(030,-28.8)(006,-8.3)
    \put(006,-8.3){\vector(-5,4){0}}

    \qbezier(-9.26,004)(-30,015)(-30,000)
    \qbezier(-9.26,-04)(-30,-15)(-30,000)
    \put(-9.26,004){\vector(2,-1){0}}
    \qbezier(-7.14,007)(-40,025)(-40,000)
    \qbezier(-7.14,-07)(-40,-25)(-40,000)
    \put(-7.14,007){\vector(3,-2){0}}

    \qbezier(69.26,004)(090,015)(090,000)
    \qbezier(69.26,-04)(090,-15)(090,000)
    \put(69.26,004){\vector(-2,-1){0}}

\put(000,000){\makebox(0,0){$\alpha$}}
\put(060,000){\makebox(0,0){$\beta$}}

\put(-43,000){\makebox(0,0)[r]{$11/cd$}}
\put(-27,000){\makebox(0,0)[l]{$0/a$}}
\put(030,011.0){\makebox(0,0)[b]{$10/\mathbf{b}$}}
\put(030,-08.0){\makebox(0,0)[b]{$1/\mathbf{da}$}}
\put(030,-20.0){\makebox(0,0)[t]{$01/\mathbf{dcd}$}}
\put(093,000){\makebox(0,0)[l]{$00/db$}}
\end{picture}
\thinlines
\setlength{\unitlength}{1pt}
\end{center}
\caption{Third variable-length encoder~$\encoder^\circ$ for
the constraint presented by Figure~\protect\ref{fig:twostates}.}
\label{fig:vleHMS}
\end{figure}
It is easily seen that while at state~$\alpha$,
each outgoing edge is uniquely determined by its first symbol,
and while at state~$\beta$, an outgoing edge is
uniquely determined by its first two symbols.

\ifIEEE
   \appendices
\else
   \section*{$\,$\hfill Appendix\hfill$\,$}
   \appendix
\fi

\section{Limitations of state splitting}
\label{sec:limitations}

In contrast to what we have shown in Section~\ref{sec:anticipationone},
we present here an example where the state-splitting algorithm
yields encoders~$\encoder_0$ and~$\encoder_1$ with anticipation~$3$,
yet any attempt to match between
the descendant states in~$\encoder_0$ of a given $G$-state
and the respective descendant states
in~$\encoder_1$ results in an encoder that has
no finite anticipation.\footnote{%
It is still open whether such an example exists 
where the anticipation of~$\encoder_0$ and~$\encoder_1$ is~$2$.
It is not difficult to construct such an example where
\emph{some} matchings of the descendant states in~$\encoder_0$
with those in~$\encoder_1$ of
the same $G$-state yield an encoder with infinite anticipation.}

Let $G$ be the graph with
$V(G) = \{ \alpha, \beta, \gamma, \delta \}$
whose even and odd subgraphs,
$G_0$ and $G_1$, are shown in Figures~\ref{fig:G0} and~\ref{fig:G1}:
the even-valued (\resp, odd-valued)
hexadecimal digits form the set~$\Sigma_0$ (\resp, $\Sigma_1$).
Note that~$G_0$ and~$G_1$ are identical graphs
except for the edge labeling and for
a switch between states~$\gamma$ and~$\delta$.

\newsavebox{\countergraph}
\sbox{\countergraph}{%
    \thicklines
    \setlength{\unitlength}{0.48mm}
    \multiput(000,000)(060,000){4}{\circle{20}}
    \qbezier(051,4.5)(030,14.5)(009,4.5)
    \put(051,4.5){\vector(2,-1){0}}
    \qbezier(051,-4.5)(030,-14.5)(009,-4.5)
    \put(009,-4.5){\vector(-2,1){0}}
    \put(070,000){\vector(1,0){40}}
    \qbezier(171,4.7)(150,14.5)(129,4.7)
    \put(171,4.7){\vector(2,-1){0}}
    \qbezier(171,-4.7)(150,-14.5)(129,-4.7)
    \put(171,-4.7){\vector(2,1){0}}
    \qbezier(172.0,006)(120,046)(068.0,006)
    \put(068.0,006){\vector(-3,-2){0}}
    \qbezier(173.5,7.8)(090,68.8)(006.5,7.8)
    \put(006.5,7.8){\vector(-3,-2){0}}
    \qbezier(189.26,004)(210,015)(210,000)
    \qbezier(189.26,-04)(210,-15)(210,000)
    \put(189.26,-04){\vector(-2,1){0}}
}

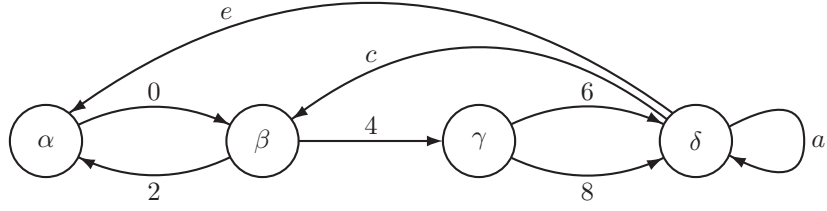
\begin{figure*}[hbt]
\begin{center}
\thicklines
\setlength{\unitlength}{0.48mm}
\begin{picture}(225,65)(-10,-15)
\put(000,000){\usebox{\countergraph}}
\put(000,000){\makebox(0,0){$\alpha$}}
\put(060,000){\makebox(0,0){$\beta$}}
\put(120,000){\makebox(0,0){$\gamma$}}
\put(180,000){\makebox(0,0){$\delta$}}

\put(030,011.5){\makebox(0,0)[b]{$0$}}
\put(030,-11.5){\makebox(0,0)[t]{$2$}}
\put(090,002){\makebox(0,0)[b]{$4$}}
\put(150,011.5){\makebox(0,0)[b]{$6$}}
\put(150,-11.5){\makebox(0,0)[t]{$8$}}
\put(212,000){\makebox(0,0)[l]{$a$}}
\put(090,022){\makebox(0,0)[b]{$c$}}
\put(050,034){\makebox(0,0)[b]{$e$}}

\end{picture}
\thinlines
\setlength{\unitlength}{1pt}
\end{center}
\caption{Subgraph $G_0$.}
\label{fig:G0}
\end{figure*}

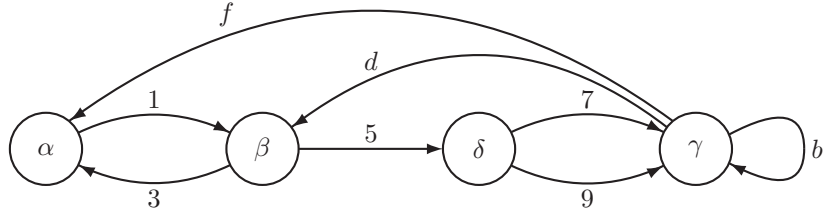
\begin{figure*}[hbt]
\begin{center}
\thicklines
\setlength{\unitlength}{0.48mm}
\begin{picture}(225,65)(-10,-15)
\put(000,000){\usebox{\countergraph}}
\put(000,000){\makebox(0,0){$\alpha$}}
\put(060,000){\makebox(0,0){$\beta$}}
\put(120,000){\makebox(0,0){$\delta$}}
\put(180,000){\makebox(0,0){$\gamma$}}

\put(030,011.5){\makebox(0,0)[b]{$1$}}
\put(030,-11.5){\makebox(0,0)[t]{$3$}}
\put(090,002){\makebox(0,0)[b]{$5$}}
\put(150,011.5){\makebox(0,0)[b]{$7$}}
\put(150,-11.5){\makebox(0,0)[t]{$9$}}
\put(212,000){\makebox(0,0)[l]{$b$}}
\put(090,022){\makebox(0,0)[b]{$d$}}
\put(050,034){\makebox(0,0)[b]{$f$}}
\end{picture}
\thinlines
\setlength{\unitlength}{1pt}
\end{center}
\caption{Subgraph~$G_1$.}
\label{fig:G1}
\end{figure*}
Assuming the ordering $\alpha < \beta < \gamma < \delta$
on the states, the adjacency matrices of the subgraphs are given by
\[
A_{G_0} =
\left(
\arraycolsep0.7ex
\begin{array}{cccc}
0 & 1 & 0 & 0 \\
1 & 0 & 1 & 0 \\
0 & 0 & 0 & 2 \\
1 & 1 & 0 & 1
\end{array}
\right)
\quad \textrm{and} \quad
A_{G_1} =
\left(
\arraycolsep0.7ex
\begin{array}{cccc}
0 & 1 & 0 & 0 \\
1 & 0 & 0 & 1 \\
1 & 1 & 1 & 0 \\
0 & 0 & 2 & 0
\end{array}
\right)
\; ,
\]
and it is easily seen that
$\bldx = (1 \; 2 \; 3 \; 3)^\top$ is an eigenvector of both matrices
associated with the Perron eigenvalue $n = 2$.
We are interested in constructing an~$(S(G),2,2)$-encoder.

Applying a first round of $\bldx$-consistent state splitting
to $G_0$ allows splitting of state~$\delta$
(and only that state in~$G_0$) into two descendant states:
$(\delta,0)$, which has a weight (i.e., approximate eigenvector entry)
of~$1$ and inherits the outgoing edge labeled~$c$,
and $(\delta,1)$, of weight~$2$, which inherits the outgoing edges
labeled~$a$ and~$e$.

In the second round, state~$(\delta,1)$ can be split
into~$(\delta,1,0)$ and~$(\delta,1,1)$
(each of weight~$1$),
and state~$\gamma$ can be fully split into three descendant states,
$(\gamma,\cdot,0)$, $(\gamma,\cdot,1)$, and~$(\gamma,\cdot,2)$,
each of weight~$1$. The outgoing picture from the descendant
states of~$\gamma$ is shown in Figure~\ref{fig:gamma}(a),
where
\[
(\gamma_{0;1},\gamma_{0;2},\gamma_{0;3})
= ((\gamma,\cdot,0),(\gamma,\cdot,1),(\gamma,\cdot,2))
\]
and
\[
(\delta_{0;1},\delta_{0;2},\delta_{0;3})
= ((\delta,0),(\delta,1,0),(\delta,1,1))
\]
(the first subscript in~$\gamma_{0;i}$ and~$\delta_{0;j}$
indicates that we are splitting the graph~$G_0$).

\newsavebox{\gammastates}
\sbox{\gammastates}{
    \thicklines
    \setlength{\unitlength}{0.48mm}
    \multiput(000,000)(060,000){2}{
        \multiput(000,000)(000,030){3}{\circle{20}}
    }
    \put(000,060){\makebox(0,0){$\gamma_{0;1}$}}
    \put(000,030){\makebox(0,0){$\gamma_{0;2}$}}
    \put(000,000){\makebox(0,0){$\gamma_{0;3}$}}
    \put(060,060){\makebox(0,0){$\delta_{0;1}$}}
    \put(060,030){\makebox(0,0){$\delta_{0;2}$}}
    \put(060,000){\makebox(0,0){$\delta_{0;3}$}}
}
\newsavebox{\deltastates}
\sbox{\deltastates}{
    \thicklines
    \setlength{\unitlength}{0.48mm}
    \multiput(000,000)(060,000){2}{
        \multiput(000,000)(000,030){3}{\circle{20}}
    }
    \put(000,060){\makebox(0,0){$\delta_{1;1}$}}
    \put(000,030){\makebox(0,0){$\delta_{1;2}$}}
    \put(000,000){\makebox(0,0){$\delta_{1;3}$}}
    \put(060,060){\makebox(0,0){$\gamma_{1;1}$}}
    \put(060,030){\makebox(0,0){$\gamma_{1;2}$}}
    \put(060,000){\makebox(0,0){$\gamma_{1;3}$}}
}
\newsavebox{\gammadeltaedgesa}
\sbox{\gammadeltaedgesa}{
    \thicklines
    \setlength{\unitlength}{0.48mm}
    \put(000,060){
        \qbezier(051,4.5)(030,14.5)(009,4.5)
        \put(051,4.5){\vector(2,-1){0}}
        \qbezier(051,-4.5)(030,-14.5)(009,-4.5)
        \put(051,-4.5){\vector(2,1){0}}
    }
    \put(010,000){\vector(1,0){40}}
    \put(08.9,04.47){\vector(2,1){42.16}}
    \put(010,030){\vector(1,0){40}}
    \put(08.9,25.52){\vector(2,-1){42.16}}
}
\newsavebox{\gammadeltaedgesb}
\sbox{\gammadeltaedgesb}{
    \thicklines
    \setlength{\unitlength}{0.48mm}
    \multiput(000,000)(000,030){3}{
        \qbezier(051,4.5)(030,14.5)(009,4.5)
        \put(051,4.5){\vector(2,-1){0}}
        \qbezier(051,-4.5)(030,-14.5)(009,-4.5)
        \put(051,-4.5){\vector(2,1){0}}
    }
}
\begin{figure}[hbt]
\begin{center}
\thicklines
\setlength{\unitlength}{0.48mm}
\begin{picture}(180,095)(-10,-15)
\put(000,000){
    \put(000,000){\usebox{\gammastates}}
    \put(000,000){\usebox{\gammadeltaedgesa}}
    \put(000,060){
        \put(030,011.5){\makebox(0,0)[b]{$6$}}
        \put(030,-11.5){\makebox(0,0)[t]{$8$}}
    }

    \put(020,022){\makebox(0,0)[b]{$6$}}
    \put(015,010){\makebox(0,0)[b]{$8$}}
    
    \put(030,032){\makebox(0,0)[b]{$6$}}
    \put(030,002){\makebox(0,0)[b]{$8$}}
    \put(030,-20){\makebox(0,0)[t]{(a)}}
}

\put(100,000){
    \put(000,000){\usebox{\gammastates}}
    \put(000,000){\usebox{\gammadeltaedgesb}}
    \multiput(000,000)(000,030){3}{
        \put(030,011.5){\makebox(0,0)[b]{$6$}}
        \put(030,-07.5){\makebox(0,0)[b]{$8$}}
    }
    \put(030,-20){\makebox(0,0)[t]{(b)}}
}
\end{picture}
\thinlines
\setlength{\unitlength}{1pt}
\end{center}
\caption{Possible outgoing pictures from the descendants of
state~$\gamma$ in~$\encoder_0$.}
\label{fig:gamma}
\end{figure}
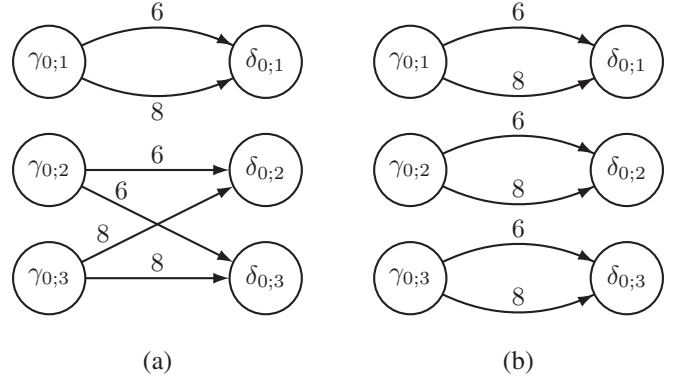

In the third round, state~$\beta$ is split into two descendant
states, $(\beta,\cdot,\cdot,0)$ and~$(\beta,\cdot,\cdot,1)$,
each of weight~$1$. At this point, we get an~$(S(G_0),2)$-encoder
$\encoder_0$.

\begin{remark}
\label{rem:limitations1}
State~$\gamma$ can alternatively be split only partially
in the second round, yielding a descendant~$(\gamma,\cdot,0)$
of weight~$1$ and
a descendant~$(\gamma,\cdot,1)$ of weight~$2$,
deferring the splitting of~$(\gamma,\cdot,1)$ to the third round.
In this case, the outgoing picture shown in
Figure~\ref{fig:gamma}(b) is also possible,
where now
$(\gamma_{0;1},\gamma_{0;2},\gamma_{0;3}) =
((\gamma,\cdot,0),(\gamma,\cdot,1,0),(\gamma,\cdot,1,1))$.\qed
\end{remark}

\begin{remark}
\label{rem:limitations2}
In~$\encoder_0$, there are three distinct paths
labeled~$0\,4$ from state~$\alpha$ to the three descendant states
($\gamma_{0;1}$, $\gamma_{0;2}$, and~$\gamma_{0;3}$) of~$\gamma$.\qed
\end{remark}

A respective splitting of~$G_1$ yields an~$(S(G_1),2)$-encoder
$\encoder_1$, in which the possible outgoing pictures from
the descendant states of~$\delta$
(namely, $\delta_{1;i}$)
are shown in Figure~\ref{fig:delta}.

\begin{figure}[hbt]
\begin{center}
\thicklines
\setlength{\unitlength}{0.48mm}
\begin{picture}(180,095)(-10,-15)
\put(000,000){
    \put(000,000){\usebox{\deltastates}}
    \put(000,000){\usebox{\gammadeltaedgesa}}
    \put(000,060){
        \put(030,011.5){\makebox(0,0)[b]{$7$}}
        \put(030,-11.5){\makebox(0,0)[t]{$9$}}
    }

    \put(020,022){\makebox(0,0)[b]{$7$}}
    \put(015,010){\makebox(0,0)[b]{$9$}}
    
    \put(030,032){\makebox(0,0)[b]{$7$}}
    \put(030,002){\makebox(0,0)[b]{$9$}}
    \put(030,-20){\makebox(0,0)[t]{(a)}}
}

\put(100,000){
    \put(000,000){\usebox{\deltastates}}
    \put(000,000){\usebox{\gammadeltaedgesb}}
    \multiput(000,000)(000,030){3}{
        \put(030,011.5){\makebox(0,0)[b]{$7$}}
        \put(030,-07.5){\makebox(0,0)[b]{$9$}}
    }
    \put(030,-20){\makebox(0,0)[t]{(b)}}
}
\end{picture}
\thinlines
\setlength{\unitlength}{1pt}
\end{center}
\caption{Possible outgoing pictures from the descendants of
state~$\delta$ in~$\encoder_1$.}
\label{fig:delta}
\end{figure}
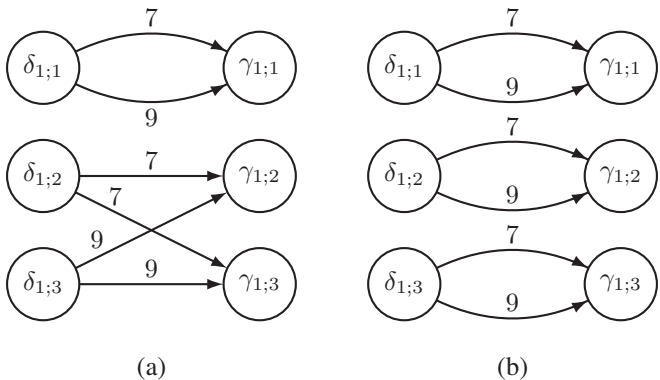

Consider now the $(S(G),2,2)$-encoder $\encoder$
obtained by matching the descendant states in $\encoder_0$ of
each given $G$-state with the respective descendant states
in $\encoder_1$. In particular, we select some bijections
\begin{equation}
\label{eq:varphi0}
\varphi_0 :
\left\{
\gamma_{1;1},
\gamma_{1;2},
\gamma_{1;3}
\right\}
\rightarrow
\left\{
\gamma_{0;1},
\gamma_{0;2},
\gamma_{0;3}
\right\}
\; \phantom{.}
\end{equation}
and
\begin{equation}
\label{eq:varphi1}
\varphi_1 :
\left\{
\delta_{0;1},
\delta_{0;2},
\delta_{0;3}
\right\}
\rightarrow
\left\{
\delta_{1;1},
\delta_{1;2},
\delta_{1;3}
\right\}
\; .
\end{equation}

Next, we show that for every such selection, there is
an arbitrarily long word $\bldw$ 
that can be generated in~$\encoder$ from two different
descendant states of~$\gamma$ (in~$\encoder_0$).
And since both these states
are reachable in~$\encoder_0$ from state~$\alpha$ by
paths labeled~$0\,4$, it will follow that~$\encoder$ does not have
finite anticipation.

Such a word~$\bldw$ will be generated by paths that toggle
between~$\encoder_0$ and~$\encoder_1$ after each symbol.
For example, suppose that
$\varphi_0(\gamma_{1;i}) = \gamma_{0;i}$ and
$\varphi_1(\delta_{0;i}) = \delta_{1;i}$ , for $i = 1, 2, 3$.
Then the word~$6\,7\,6\,7\,\ldots$ can be generated by
the following two paths:
\[
\gamma_{0;1}
\;\; \stackrel{6}{\longrightarrow}
\;\; {\delta_{0;1} \atop \delta_{1;1}} 
\;\; \stackrel{7}{\longrightarrow}
\;\; {\gamma_{0;1} \atop \gamma_{1;1}} 
\;\; \stackrel{6}{\longrightarrow}
\;\; {\delta_{0;1} \atop \delta_{1;1}} 
\;\; \stackrel{7}{\longrightarrow}
\;\; \cdots
\]
and
\[
\gamma_{0;2}
\;\; \stackrel{6}{\longrightarrow}
\;\; {\delta_{0;2} \atop \delta_{1;2}} 
\;\; \stackrel{7}{\longrightarrow}
\;\; {\gamma_{0;2} \atop \gamma_{1;2}} 
\;\; \stackrel{6}{\longrightarrow}
\;\; {\delta_{0;2} \atop \delta_{1;2}} 
\;\; \stackrel{7}{\longrightarrow}
\;\; \cdots
\]

\begin{lemma}
For any two bijections $\varphi_0$ and $\varphi_1$
as in~(\ref{eq:varphi0})--(\ref{eq:varphi1})
and for any positive integer~$\ell$,
there are at least two paths of length~$\ell$
in~$\encoder$ that satisfy the following properties.
{\parskip0ex%
\renewcommand{\theenumi}{\roman{enumi}}%
\renewcommand{\labelenumi}{(\theenumi)}%
\begin{enumerate}
\itemsep0ex
\item
\label{item:1}
The paths generate the same word.
\item
\label{item:2}
The paths start at distinct descendants of~$\gamma$ in~$\encoder_0$.
\item
\label{item:3}
The states along each path alternate between descendants of~$\gamma$
in~$\encoder_0$ and descendants of~$\delta$ in~$\encoder_1$.
\end{enumerate}
}
\end{lemma}

\begin{proof}
When both subgraphs~$G_0$ and~$G_1$ are split according to
part~(b) in Figures~\ref{fig:gamma} and~\ref{fig:delta},
then the word $6 \, 7 \, 6 \, 7 \, \ldots$ can be generated
in~$\encoder$ from $\gamma_{0;1}$, $\gamma_{0;2}$, and~$\gamma_{0;3}$.
Hence, we assume from now on in the proof that at most
one of the subgraphs is split according to part~(b).

Our proof is by induction on~$\ell$.
The case~$\ell = 1$ is obvious, yet when~$G_0$ is split
according to Figure~\ref{fig:gamma}(b), we will also
need to establish the case~$\ell = 2$.
Let $i, j$ be distinct in~$\{ 1, 2, 3 \}$ such that
\[
\left\{
\varphi_1(\delta_{0;i}), \varphi_1(\delta_{0;j})
\right\}
\ne
\left\{ \delta_{1;2}, \delta_{1;3} \right\}
\; .
\]
By Figure~\ref{fig:delta}, this selection guarantees that
$\varphi_1(\delta_{0;i})$ and $\varphi_1(\delta_{0;j})$
share an outgoing label~$w' \in \{ 7, 9 \}$. 
Hence, if~$G_0$ is split according to part~(b),
then the word $6 \, w'$ (as well as the word~$8 \, w'$) can be
generated in~$\encoder$ both from~$\gamma_{0;i}$
and from~$\gamma_{0;j}$.

Turning to the induction step, assume that for some odd positive~$\ell$
there exist paths $\pi_1$ and~$\pi_2$
that satisfy properties~(\ref{item:1})--(\ref{item:3}),
and let~$\bldw$ be the word generated by both paths;
due to the symmetry between
Figures~\ref{fig:gamma} and~\ref{fig:delta}, the proof is
similar for even~$\ell$.
Let $\gamma_{0;i}$ (\resp, $\gamma_{0;j}$) be
the \emph{penultimate} state visited along the path~$\pi_1$
(\resp, $\pi_2$);
note that $i \ne j$, or else~$\encoder$ would not be lossless
(by Remark~\ref{rem:limitations1}).
We now distinguish between two cases.

\emph{Case~1: $G_0$ is split according to Figure~\ref{fig:gamma}(a).}
In that figure
$\gamma_{0;2}$ and~$\gamma_{0;3}$ do not share any outgoing
labels and, therefore, $\{ i, j \} \ne \{ 2, 3 \}$.
Thus, $i$ (say) is~$1$---%
and therefore $\pi_1$ terminates in~$\delta_{0;1}$---%
and $j \in \{ 2, 3 \}$, and, without loss of generality,
$\pi_2$ terminates in~$\delta_{0;2}$; moreover, there
exists a path~$\pi_2'$ that differs from~$\pi_2$ only
in that it terminates in~$\delta_{0;3}$ instead
(and~$\pi_2'$ still generates the same word~$\bldw$).
Now, any descendant state of~$\delta$ in~$\encoder_1$ shares
at least one outgoing label with at least one other descendant state
of~$\delta$ in $\encoder_1$; hence,
$\varphi_1(\delta_{0;1})$ must have a common outgoing label
$w' \in \{ 7, 9 \}$ with either $\varphi_1(\delta_{0;2})$
or $\varphi_1(\delta_{0;3})$.
Thus, $\pi_1$, as well as either $\pi_2$ or~$\pi_2'$,
can be extended by an edge (of $\encoder_1$) labeled $w'$,
to produce two paths of length~$\ell{+}1$
that satisfy properties~(\ref{item:1})--(\ref{item:3})
(both paths generating the word $\bldw w'$).

\emph{Case~2: $G_0$ is split according to Figure~\ref{fig:gamma}(b).}
By our assumption this implies that~$G_1$ is split according
to Figure~\ref{fig:delta}(a), so we can apply the analysis
of Case~1 to length~$\ell{-}1$ (with~$G_0$ and~$G_1$ switching roles).
In particular, there exist paths~$\pi_1$, $\pi_2$, and $\pi_2'$
of length~$\ell{-}1$, all generating the same word~$\bldw$,
such that~$\pi_1$ and~$\pi_2$ start at distinct descendants
of $\gamma$ (in~$\encoder_0$),
$\pi_1$ terminates in~$\gamma_{1;1}$,
and~$\pi_2$ and~$\pi_2'$ differ only
in their last edge: $\pi_2$ terminates in~$\gamma_{1;2}$
while~$\pi_2'$ terminates in~$\gamma_{1;3}$.
Write $\gamma_{0;r} = \varphi_0(\gamma_{1;1})$,
$\gamma_{0;s} = \varphi_0(\gamma_{1;2})$, and
$\gamma_{0;t} = \varphi_0(\gamma_{1;3})$.
Similarly to Case~1,
$\varphi_1(\delta_{0;r})$ must have a common outgoing label
$w' \in \{ 7, 9 \}$ with either $\varphi_1(\delta_{0;s})$
or $\varphi_1(\delta_{0;t})$.
Thus, $\pi_1$, as well as either $\pi_2$ or~$\pi_2'$,
can be extended by two edges---the first of~$\encoder_0$
labeled $6$ (or~$8$) and the second of~$\encoder_1$ labeled~$w'$---%
to produce two paths of length~$\ell{+}1$
that satisfy properties~(\ref{item:1})--(\ref{item:3})
(both paths generating either the word $\bldw \, 6 \, w'$
or the word $\bldw \, 8 \, w'$).
\end{proof}

\end{document}